\RequirePackage{amsmath}
\documentclass{llncs}
\usepackage{amsmath}
\usepackage{amssymb}
\usepackage{url}
\usepackage[bookmarks,bookmarksdepth=2,pdfusetitle]{hyperref}
\usepackage{graphicx,multirow,wasysym}
\usepackage{lmodern}
\usepackage{amsfonts,amsmath}
\usepackage{paralist}
\usepackage{enumerate}
\usepackage{bbm}
\usepackage{tikz}
\usepackage{verbatim}
\usepackage[shortlabels]{enumitem}
\usepackage{comment}
\usepackage{booktabs}
\usepackage{tabularx}
\usepackage{xmpmulti}
\usepackage{multicol}
\usepackage{float}
\floatstyle{boxed}
\restylefloat{figure}
\begin{document}



\newcommand{\mbbN}{\mathbb{N}}
\newcommand{\mcA}{\mathcal{A}}
\newcommand{\mcB}{\mathcal{B}}
\newcommand{\mcS}{\mathcal{S}}
\newcommand{\mcD}{\mathcal{D}}

\newcommand{\mb}{\mathbf}
\newcommand{\ms}{\mathsf}
\newcommand{\mbb}{\mathbb}
\newcommand{\mc}{\mathcal}

\pagestyle{plain}
\mainmatter
\title{Forward-Secure Group Signatures from Lattices}
\author{San Ling, Khoa Nguyen, Huaxiong Wang, Yanhong Xu}
\institute{
Division of Mathematical Sciences, \\
School of Physical and Mathematical Sciences,\\
Nanyang Technological University, Singapore.\\
\texttt{\{lingsan,khoantt,hxwang,xu0014ng\}@ntu.edu.sg}
}
\maketitle
\begin{abstract}

Group signature is a fundamental cryptographic primitive, aiming to protect anonymity and ensure accountability of users. It allows group members to anonymously sign messages on behalf of the whole group, while incorporating a tracing mechanism to identify the signer of any suspected signature. Most of the existing group signature schemes, however, do not guarantee security once secret keys are exposed. To reduce potential damages caused by key exposure attacks, Song (ACMCCS 2001) put forward the concept of forward-secure group signature (FSGS), which prevents attackers from forging group signatures pertaining to past time periods even if a secret group signing key is revealed at the current time period. For the time being, however, all known secure FSGS schemes are based on number-theoretic assumptions, and are vulnerable against quantum computers.

~~In this work, we construct the first lattice-based FSGS scheme. Our scheme is proven secure under the Short Integer Solution   and Learning  With Errors  assumptions. At the heart of our construction is a scalable lattice-based key evolving mechanism, allowing users to periodically update their secret keys and to efficiently prove in zero-knowledge that key evolution process is done correctly. To realize this essential building block, we first employ the Bonsai tree structure by Cash et al. (EUROCRYPT 2010) to handle the key evolution process, and then develop Langlois et al.'s construction (PKC 2014) to design its supporting zero-knowledge protocol.

\smallskip
{\bf{Keywords.}} Group signatures, key exposure, forward-security, lattice-based cryptography, zero-knowledge proofs
\end{abstract}

\section{Introduction}
\noindent{\sc{Group signatures.}} Initially suggested by Chaum and van Heyst~\cite{CV91}, group signature (GS) allows users of a group controlled by a manager to sign messages anonymously in the name of the group (anonymity). Nevertheless, there is a tracing manager to identify the signer of any signature should the user abuse the anonymity (traceability). These seemingly contractive features, however, allow group signatures to find applications in various real-life
scenarios such as e-commence systems and anonymous online communications. Unfortunately, the exposure of group signing keys renders almost all the existing schemes unsatisfactory in practice. Indeed, in the traditional models of group signatures, e.g.,~\cite{BMW03,KTY04,BS04VLR,BSZ05,KY06,SEHKMO12}, the security of the scheme is no longer guaranteed when the key exposure arises. So now let us look closely at the key exposure problem  and the countermeasures to it.


 \smallskip
  \noindent
{\sc Exposure of Group Signing Keys and Forward-Secure Group Signatures.} Exposure of users' secret keys  is one of the greatest dangers to many cryptographic protocols in practice~\cite{Song01}. Forward-secure mechanisms first introduced by Anderson~\cite{Anderson02}, aim to minimize the damages caused by secret key exposures. More precisely, forward-security protects past uses of private keys in earlier time periods even if a break-in occurs currently. Afterwards, many forward-secure schemes were constructed, such as forward-secure signatures~\cite{BM99,AR00,IR01}, forward-secure public key encryption systems~\cite{DKXY02,CHK03,BBG05}, and forward-secure signatures with un-trusted update~\cite{BSSW06,LQY07,LQY10}. 
At the heart of these schemes is a key evolving technique that operates as follows. It divides the lifetime of the scheme into discrete $T$ time periods. Upon entering a new time period, a subsequent secret key is  computed from the current one via a \emph{one-way} key evolution algorithm. Meanwhile, the current key is deleted promptly. Due to the one-wayness of the updating algorithm, the security of the previous keys is preserved even though the current one is compromised. 
Therefore, by carefully choosing a secure scheme that operates well with a  key evolving mechanism, forward-security of the scheme can be guaranteed.

As investigated by Song~\cite{Song01}, secret key exposure in group signatures is much more damaging than in  ordinary digital signatures.
In group signatures, if one group member's signing key is disclosed to the attacker, then the latter can sign arbitrary messages. 
In this situation, if the underlying group signature scheme is not secure against exposure of group signing keys, then the whole system has to be re-initialized, which is obviously inefficient in practice. Besides its inefficiency, this solution is also unsatisfactory. Once there is a break-in of the system, all previously signed group signatures become invalid since we do not have a mechanism to distinguish whether a signature is generated by a legitimate group member or by the attacker. What is worse, one of the easiest way for a misbehaving member Eve to attack the system and/or to repudiate her illegally signed signatures is to reveal her group signing key secretly in the Internet and then claim to be a victim of the key exposure problem~\cite{IR01}. Now the users who had accepted signatures \emph{before} Eve's group signing key is exposed are now at the mercy of all the group members, some of whom (e.g., Eve) would not reissue the signatures with the new key.

The aforementioned problems induced by the exposure of group signing keys motivated Song~\cite{Song01} to put forward  the notion of forward-secure group signature (FSGS), in which group members are able to update their group signing keys at each time period via a one-way key evolution algorithm. Therefore, when some group member's singing key is disclosed, all the signatures generated during past periods remain valid, which then prevents dishonest group members from repudiating signatures by simply exposing keys.
Later, Nakanishi, Hira, Funabiki~\cite{NHF09} defined a rigourous security model of FSGS for static groups, where users are fixed throughout the scheme, and demonstrated a pairing-based construction. Subsequently, Libert and Yung~\cite{LY10} extended Nakanishi et al.'s work to capture the setting of the dynamically growing groups. However, all these  schemes are constructions based on number-theoretic assumptions and are fragile in the presence of quantum adversaries. In order not to put all eggs in one basket, it is imperative to consider instantiations based on alternative, post-quantum foundations, e.g., lattice assumptions. In view of this, let us now look at the topic of lattice-based group signatures.


 \smallskip

 \noindent
{\sc Lattice-based group signatures.}
In 2010, Gordon et al.~\cite{GKV10} introduced the first lattice-based instantiation of GS. Since then, numerous schemes have been put forward with various improvements on security, efficiency, and functionality. While many of them~\cite{CNR12,LLLS13,LNW15,NZZ15,LLNW16,BCN18,PinoLS18} aim to provide enhancement on security and efficiency, they are solely designed for the static groups and often fall too short for specific needs of real-life applications. With regard to advanced features, there have been proposed several schemes~\cite{LLNW14-PKC,LLMNW16-dgs,LNWX17,LNWX18,LMN16,LNWX19-ats} and they are still behind their counterparts in the number-theoretic setting. Specifically, \cite{LLNW14-PKC,LLMNW16-dgs,LNWX17,LNWX18} deal with dynamic user enrollments and/or revocations of misbehaving users while~\cite{LMN16,LNWX19-ats} attempt to restrict the power of the tracing manager or keep his actions accountable. For the time being, the problem of making GS secure against the key exposure problem is still open in the context of lattices. Taking into account the great threat of key exposure to GS and the vulnerability of GS from number-theoretic assumptions in front of quantum computers, it would be tempting to investigate lattice-based instantiations of FSGS. Furthermore, it would be desirable to achieve it with reasonable overhead, e.g., with complexity at most poly-logarithmic in~$T$.

\medskip
\noindent
{\sc{Our Contributions.}} We introduce the first FSGS scheme in the context of lattices. 
The scheme satisfies the security requirements put forward by Nakanishi et al.~\cite{NHF09} in the random oracle model. Assuming the hardness of the Short Integer Solution ($\mathsf{SIS}$) problem and  the Learning With Errors ($\mathsf{LWE}$) problem, our scheme achieves full anonymity and a stronger notion of traceability named forward-secure traceability, which captures the traceability in the setting of key exposure problems.
Let $\lambda$ be the security parameter, $N$ be the expected number of group members, and $T$ be total time periods, our construction achieves signature  size $\widetilde{\mathcal{O}}(\lambda(\log N+\log T))$, group public key  size $\widetilde{\mathcal{O}}(\lambda^2(\log N+\log T))$, and secret key  size $\widetilde{\mathcal{O}}(\lambda^2(\log N+\log T)^2\log T)$. In particular, forward security is achieved with a reasonable cost: the size of keys and signatures are at most $\mathcal{O}(\log^3 T)$ larger than those of the basic GS scheme~\cite{LLNW14-PKC} upon which we build ours.

\smallskip
\noindent
{\sc{Overview of Our Techniques.}} Typically, designing secure GS requires a  combination of  digital signature, encryption scheme and zero-knowledge ($\mathsf{ZK}$) protocol.
Let us first consider an ordinary GS scheme similar to the template proposed by Bellare et al.~\cite{BMW03}.  In the scheme, each user is assigned an $\ell$ bit string $\mathrm{id}$ as identity, where $\ell=\log N$. The user's signing key is a signature on his identifier $\mathrm{id}$, generated by the group manager.  Specifically, we let the signing key be a short vector $\mathbf{v}_{\mathrm{id}}$  satisfying $\mathbf{A}_{\mathrm{id}}\cdot\mathbf{v}_{\mathrm{id}}=\mathbf{u} \bmod q$ for some public vector $\mathbf{u}$. When signing a message, the user first encrypts his identity $\mathrm{id}$ to a ciphertext $ \mathbf{c}$ and proves that he possesses a valid signature on his identity that is also correctly encrypted to $\mathbf{c}$.
To achieve forward-security, we would need a mechanism to update the group signing key periodically and a $\mathsf{ZK}$ protocol to prove that the key updating procedure is done honestly.


Inspired by the HIBE-like key evolving technique from Nakanishi et al.~\cite{NHF09} and Libert and Yung~\cite{LY10}, which in turn follows from~\cite{CHK03,BBG05,BSSW06}, we exploit the   hierarchical structure of the Bonsai tree~\cite{CHKP10} to enable periodical key updating. 
To the best of our knowledge, this is the only lattice-based HIBE in the standard model with supporting (Stern-like~\cite{Ste96}) $\mathsf{ZK}$ proofs by Langlois et al.~\cite{LLNW14-PKC}, which seems to be the right stepping stone towards our goal.
Let $T=2^d$ be the total number of time periods. To enable key updating, each user $\mathrm{id}$ is associated with a subtree of depth $d$, where the leaves of the tree correspond to successive time periods in the apparent way.
Let the subtree be identified by matrices $\mathbf{A}_{\mathrm{id}},\mathbf{A}_{\ell+1}^{0},\mathbf{A}_{\ell+1}^{1},\ldots,\mathbf{A}_{\ell+d}^{0},\mathbf{A}_{\ell+d}^1$ and
%
%
$z=\mathsf{Bin}(t)$ be the binary representation of $t$. In order to show the key evolution is done correctly, we observe that it suffices to prove possession of  a (short) Bonsai signature $\mathbf{v}_{\mathrm{id}\|z}$  satisfying $[\mathbf{A}_{\mathrm{id}}|\mathbf{A}_{\ell+1}^{z[1]}|\cdots|\mathbf{A}_{\ell+d}^{z[d]}]\cdot \mathbf{v}_{\mathrm{id}\|z}=\mathbf{u}\bmod q$. %
%
However, proving knowledge of the Bonsai signature departs from the protocol presented in~\cite{LLNW14-PKC}. The matrix $\mathbf{A}_{\mathrm{id}}$ should be secret  and the binary string $z$ should be public in our case while it is the other way around in~\cite{LLNW14-PKC}. Nevertheless, analyzing the above equation carefully, it actually reduces to proving knowledge of short vectors $\mathbf{w}_1$ and $\mathbf{w}_2$ and a binary string $\mathrm{id}$ such that  $\mathbf{A}_{\mathrm{id}}\cdot\mathbf{w}_1+\mathbf{A}''\cdot\mathbf{w}_2=\mathbf{u}\bmod q$,  where $\mathbf{v}_{\mathrm{id}\|z}=(\mathbf{w}_1\|\mathbf{w}_2)$ and $\mathbf{A}''$ is built from some public matrices. To prove knowledge of $\mathbf{w}_2$, we can employ the decomposition/extension/permutation techniques by Ling et al.~\cite{LNSW13} that operate in  Stern's framework~\cite{Ste96}. Regarding the $\mathsf{ZK}$ protocol for proving knowledge of $\mathbf{w}_1$ and $\mathrm{id}$, it indeed depends on   the signature scheme used by the group manager to certify users. For simplicity, we employ the Bonsai tree signature~\cite{CHKP10} as well. Then, by utilizing the $\mathsf{ZK}$ protocol in~\cite{LLNW14-PKC}, we are able to prove knowledge of $\mathbf{w}_1$ and $\mathrm{id}$ and manage to obtain the desired $\mathsf{ZK}$ protocol for proving possession of $\mathbf{v}_{\mathrm{id}\|z}$.  It is worth mentioning  that, besides the Bonsai signature, the Boyen signature~\cite{Boy10} is also a plausible candidate, for which a $\mathsf{ZK}$ protocol showing the possession of  a valid message-signature pair was known~\cite{LNW15}.

In the above, we have discussed the  (Stern-like) $\mathsf{ZK}$ protocol showing knowledge of correctly updated signing key $\mathbf{v}_{\mathrm{id}\|\mathsf{Bin}(t)}$, the main technical building block in achieving our FSGS scheme.
The next question is then how should the user derive $\mathbf{v}_{\mathrm{id}\|\mathsf{Bin}(t)}$ for all possible $t$ using his group signing key $\mathbf{v}_{\mathrm{id}}$.  To this end, we make a minor but significant change to the group signing key. Observe that for the Bonsai tree signature, once a trapdoor matrix $\mathbf{S}_{\mathrm{id}}$ satisfying $\mathbf{A}_{\mathrm{id}}\cdot\mathbf{S}_{\mathrm{id}}=\mathbf{0} \bmod q$ is known,  the user $\mathrm{id}$ is able to generate $\mathbf{v}_{\mathrm{id}\|\mathsf{Bin}(t)}$ for all possible $t$. Therefore, we let the user's signing key be $\mathbf{S}_{\mathrm{id}}$ instead.
%
%
%
Nevertheless, we then observe user $\mathrm{id}$ should not hold $\mathbf{S}_{\mathrm{id}}$ at all times, as the adversary could also generate all possible $\mathbf{v}_{\mathrm{id}\|\mathsf{Bin}(t)}$ once $\mathbf{S}_{\mathrm{id}}$ is known to him. One trivial method is to generate all possible $\mathbf{v}_{\mathrm{id}\|\mathsf{Bin}(t)}$ and then delete all the previous ones upon entering a new period. However, this will incur linear dependency on $T$, which is undesirable for efficiency purpose.

To achieve logarithmic overhead, we should think of a way to employ the structure of the Bonsai tree.
Let $\mathrm{Nodes}_{(t\rightarrow T-1)}$ be the set of nodes such that it has size at most $\log T$ and contains exactly one ancestor of each leaf or the leaf itself between $t$ and $T-1$\footnote{This set can be determined by the Nodeselect algorithm presented by Libert and Yung~\cite{LY10}.}. Now we let the signing key of user $\mathrm{id}$ at time $t$ be trapdoor matrices $ \mathbf{S}_{\mathrm{id}\|z} $ for all $z\in\mathrm{Nodes}_{(t\rightarrow T-1)}$. The user is then able to produce all possible $\mathbf{v}_{\mathrm{id}\|\mathsf{Bin}(t)}$  by employing $ \mathbf{S}_{\mathrm{id}\|z} $ if $z$ is an ancestor of $\mathsf{Bin}(t)$. More importantly, for each $z'\in\mathrm{Nodes}_{(t+1\rightarrow T-1)}$, there exists a unique ancestor $z\in\mathrm{Nodes}_{(t\rightarrow T-1)}$, which enables the evolving of the signing key from time  $t$ to $t+1$, thanks to the basis delegation algorithm of the Bonsai signature.

As discussed so far, we have shown how to update the key periodically and identified the $\mathsf{ZK}$ protocol for the honest behaviour of update. The thing that remains  is to find a public key encryption $(\mathsf{PKE})$ scheme that is compatible with the above ingredients. Furthermore, to achieve full anonymity, it typically requires the $\mathsf{PKE}$ scheme to be CCA-secure. To this end, we apply the CHK transform~\cite{CHK04} to the identity-based encryption scheme~\cite{GPV08}. For the obtained $\mathsf{PKE}$ scheme, we observe that there exists a Stern-like $ \mathsf{ZK}$ protocol (see~\cite{LNW15}) for proving knowledge of the plaintext, which is compatible in our setting.
To summarize, we have obtained a  lattice-based FSGS scheme by developing several technical building blocks from previous works in a non-trivial way.
Our scheme satisfies full anonymity due to the facts that the underlying encryption scheme is $\mathrm{CCA}$-secure and that the underlying $\mathsf{ZK}$ protocol is statistically zero-knowledge, and  achieves forward-secure traceability due to the security of the Bonsai tree signature~\cite{CHKP10}. We believe that, our construction - while not being truly novel - would certainly help to enrich the area of lattice-based GS. 

\smallskip

\noindent
{\sc Related Work.} Recently,  Kansal, Dutta and  Mukhopadhyay~\cite{KDM17} proposed a lattice-based FSGS scheme  that operates in the model of Libert and Yung~\cite{LY10}. Unfortunately, it can be observed that their construction does not satisfy the correctness and security requirements of~\cite{LY10}. 
(For details, see Appendix~\ref{appendix:incorrect}.)

\section{Preliminaries}\label{section:prelim}

\noindent Throughout the paper, all vectors are column vectors. When concatenating two matrices of form $\mathbf{A}\in\mathbb{R}^{n\times m}$ and $\mathbf{B}\in\mathbb{R}^{n\times k}$, we use the notion $[\mathbf{A}|\mathbf{B}]\in\mathbb{R}^{n\times (m+k)}$ while we denote $(\mathbf{x}\|\mathbf{y})\in\mathbb{R}^{m+k}$ as  the concatenation of two vectors of form $\mathbf{x}\in\mathbb{R}^m$ and $\mathbf{y}\in\mathbb{R}^k$. 
Let $[m]$ be the set $\{1,2,\cdots, m\}$.

\subsection{Forward-Secure Group Signatures}\label{subsection: model and security}
We now recall the syntax and security requirements of forward-secure group signature (FSGS), as formalized by Nakanishi et al.~\cite{NHF09}.
An FSGS scheme consists of the following polynomial-time algorithms.

\begin{description}
\item $\mathsf{KeyGen}$: This algorithm takes the tuple $(\lambda, T, N)$ as input, with $\lambda$ being security parameter, $T$ being total number of time periods, and $N$ being maximum number of  group members. It then returns group public key $\mathsf{gpk}$,  secret key $\mathsf{msk}$ of  group manager ($\mathsf{GM}$),  secret key $\mathsf{mosk}$ of tracing manager ($\mathsf{TM}$), initial user secret keys $\mathbf{usk}_0$.  $\mathbf{usk}_0$ is an array of initial $N$ secret signing key $\{\mathbf{usk}_0[0],\mathbf{usk}_0[1],\cdots,\mathbf{usk}_0[N-1]\}$,  with $\mathbf{usk}_0[i]$ being the initial key of user~$i$.
\item $\mathsf{KeyUpdate}$:  On inputs $\mathsf{gpk},~\mathbf{usk}_{t}[i]$, $i$, and $t+1$, with $\mathbf{usk}_{t}[i]$ being the secret signing key of user $i$ at time $t$, this randomized algorithm outputs the secret signing key $\mathbf{usk}_{t+1}[i]$ of user $i$ at time $t+1$.
\item $\mathsf{Sign}$: On inputs $\mathsf{gpk}$, $\mathbf{usk}_{t}[i]$, user $i$, time period $t$, and message $M$, this randomized algorithm generates a signature $\Sigma$ on   message $M$.
\item $\mathsf{Verify}$: It takes as inputs $\mathsf{gpk}$, time period $t$, message $M$ and signature $\Sigma$, and returns $1/0$ indicating the validity of the signature.
\item $\mathsf{Open}$: On inputs $\mathsf{gpk}$, $\mathsf{mosk}$, $t$, $M$ and $\Sigma$, this deterministic algorithm returns an index $i$ or $\bot$.

\end{description}

\noindent\emph{Correctness.} For all $\lambda, T,N$,  $(\mathsf{gpk},\mathsf{msk},\mathsf{mosk},\mathbf{usk}_0)\leftarrow \mathsf{KeyGen}(\lambda,T,N)$, $\forall i\in\{0,1,\cdots,N-1\}$, all $M\in\{0,1\}^*$, all $\mathbf{usk}_t[i]\leftarrow \mathsf{KeyUpdate}(\mathsf{gpk},\mathbf{usk}_{t-1}[i],i,t)$ for all $t\in\{0,1,\cdots T-1\}$, the following equations hold:
$$\mathsf{Verify}(\mathsf{gpk},t,M,\mathsf{Sign}(\mathsf{gpk},\mathbf{usk}_t[i],t,M))=1,$$
$$\mathsf{Open}(\mathsf{gpk},\mathsf{mosk},t,M,\mathsf{Sign}(\mathsf{gpk},\mathbf{usk}_t[i],t,M))=i.$$

\noindent\emph{Forward-Secure Traceability.} This requirement demands that any $\mathrm{PPT}$ adversary, even if it can corrupt the tracing manager and some (or all) group members, is not able to produce a valid signature (i) that is opened to some non-corrupted user or (ii) that is traced to some corrupted user, but the signature is signed at time period preceding the secret key query of this corrupted user. Note that (i) captures the standard traceability requirement as in~\cite{BMW03} while (ii) deals with the new requirement in the context of forward-security. Details are modelled in the experiment in Fig.~\ref{Exp:trace}.

In the experiment in  Fig.~\ref{Exp:trace}, the adversary can adaptively choose which user to corrupt, when to corrupt and  when to halt, and when to output its forgery. Furthermore, it is allowed to query signature on any message of a member $i$ through the signing oracle $\mathsf{Sign}(\mathbf{usk}_{t}[\cdot],\cdot)$ if $i\notin \mathrm{CU}$ at time period~$t$.

Define the advantage $\mathbf{Adv}_{\mathrm{FSGS},\mathcal{A}}^{\ms{Trace}}(\lambda,T,N)$ of adversary $\mathcal{A}$ against forward-secure traceability of an FSGS scheme as $\text{Pr}[\mathbf{Exp}_{\mathrm{FSGS},\mathcal{A}}^{\mathsf{Trace}}(\lambda,T,N)=1]$. An FSGS scheme is forward-secure traceable if the advantage of any $\mathrm{PPT}$ adversary is negligible.

\begin{figure}
  \centering


    \begin{minipage}{13cm}
    \underline{$\mathrm{Experiment}~\mathbf{Exp}_{\mathrm{FSGS},\mathcal{A}}^{\mathsf{Trace}}(\lambda,T,N)$} \\
    $\>(\mathsf{gpk},\mathsf{msk},\mathsf{mosk},\mathbf{usk}_0)\leftarrow\mathsf{KeyGen}(\lambda,T,N)$,\\
    $\>\mathrm{st}\leftarrow(\mathsf{gpk},\mathsf{mosk}),~\mathsf{CU}\leftarrow\emptyset$, $t:=0$, $\mathrm{Cont}=1,~\mathrm{stop}=0,~K\leftarrow\epsilon$,\\
    $\> \mathrm{For}(t=0,~t\leq T-1,~t++)$\\
    $\>~~~\mathrm{While}(\mathrm{Cont}=1)$ \\
           $~~~~~~\>(\mathrm{Cont},j,\mathrm{st})\leftarrow\mathcal{A}^{\mathsf{Sign}(\mathbf{usk}_{t}[\cdot],\cdot)}(\mathrm{choose1},K,\mathrm{st})$,\\
           $~~~~~~\> \mathrm{if}~\mathrm{Cont}=1,~ \mathrm{let}~\mathsf{CU}\leftarrow\mathsf{CU}\cup\{j\},K\leftarrow K\cup\{\mathsf{usk}_t[j],t\}$, $\mathrm{endif}$\\
    $\>~~~\mathrm{Endwhile}$  \\
    $\>~~~\mathrm{stop}\leftarrow\mathcal{A}(\mathrm{choose2,st},t,K),$\\
    $\>~~~    \mathrm{if}~ \mathrm{stop}=1, \mathrm{break};~\mathrm{else}~\mathrm{Cont}=1.$\\
    $\>\mathrm{Endfor}$\\
    $\>(t^*,M^*,\Sigma^*)\leftarrow\mathcal{A}(\mathrm{guess},\mathrm{st},t,K)$,\\
    $\> \text{Output} ~1$ if
    \begin{itemize}
    \item $\mathsf{Verify}(\mathsf{gpk},t^*,M^*,\Sigma^*)=1,$
    \item $\Sigma^*$ is not obtained from  querying the signing oracle by $\mathcal{A}$, and
    \item Compute $i^*\leftarrow\mathsf{Open}(\mathsf{gpk},\mathsf{mosk},t^*,M^*,\Sigma^*)$, then
        \begin{itemize}
        \item either $i^*\notin\mathsf{CU}$,
        \item or $i^*\in \mathsf{CU}$, but $\mathcal{A}$ only queried $\mathbf{usk}_{t}[i^*]$ for $t> t^*$.
        \end{itemize}
    \end{itemize}

  \end{minipage}

  \caption{Experiment used to define forward-secure traceability of an FSGS scheme.}\label{Exp:trace}
  \end{figure}

\noindent\emph{Full Anonymity.} This requirement demands that any $\mathrm{PPT}$ adversary is infeasible to figure out which of two signers of its choice signed the challenged message  of its choice at time period $t$ of its choice.  Details of this requirement is modelled in the experiment in  Fig.~\ref{Exp:anon}. In the  experiment in Fig.~\ref{Exp:anon}, the adversary is  accessible to secret keys of all group users and $\mathsf{GM}$,  and can query the opening of any signature  except for the challenged one through the opening oracle $\mathsf{Open}(\mathsf{mosk},\cdot)$.

Define  $\mathbf{Adv}_{\mathrm{FSGS},\mathcal{A}}^{\ms{Anon}}(\lambda,T,N)$ of  $\mathcal{A}$ against full anonymity of an FSGS scheme as $\text{Pr}[\mathbf{Exp}_{\mathrm{FSGS},\mathcal{A}}^{\mathsf{Anon}}(\lambda,T,N)=1]$. An FSGS scheme is fully anonymous if the advantage of any $\mathrm{PPT}$ adversary $\mathcal{A}$ is negligible.

 \begin{figure}
   \begin{minipage}{13cm}
  \smallskip
  \underline{$\mathrm{Experiment}~\mathbf{Exp}_{\mathrm{FSGS},\mathcal{A}}^{\mathsf{Anon}}(\lambda,T,N)$} \\
  $\>(\mathsf{gpk},\mathsf{msk},\mathsf{mosk},\mathbf{usk}_0)\leftarrow\mathsf{KeyGen}(\lambda,T,N)$,\\
  $(\mathrm{st},t,i_0,i_1,M)\leftarrow\mathcal{A}^{\mathsf{Open}(\mathsf{mosk},\cdot)}(\mathrm{choose},\mathsf{gpk},\mathsf{msk},\mathbf{usk}_0)$,\\ $b\leftarrow\{0,1\}, \Sigma\leftarrow\mathsf{Sign}(\ms{gpk},\ms{usk}_{t}[i_b],t,M)$,\\
  $b'\leftarrow\mathcal{A}^{\mathsf{Open}(\mathsf{mosk},\cdot)}(\mathrm{guess},\mathrm{st},\Sigma)$,\\
  $\mathrm{If}~b'=b,~\mathrm{then~output~1}$,\\
  $\mathrm{Else~output~0}$.

  \end{minipage}
  \caption{Experiment used to define full anonymity of an FSGS scheme.}\label{Exp:anon}
  \end{figure}

\subsection{Some Background on Lattices}
Let $n\in\mathbb{Z}^{+}$ and $\Lambda$ be a lattice of dimension $n$ over $\mathbb{R}^n$. Let $\mathbf{S}=\{\mathbf{s}_1,\cdots,\mathbf{s}_n\}\subset \mathbb{R}^n$ be a basis of $\Lambda$. For simplicity, we  write $\mathbf{S}=[\mathbf{s}_1|\cdots|\mathbf{s}_n]\in\mathbb{R}^{n\times n}$.
Define $\|\mathbf{S}\|=\mathrm{Max}_{i}\|\mathbf{s}_i\|$.
Let $\widetilde{\mathbf{S}}=[\widetilde{\mathbf{s}}_1|\cdots|\widetilde{\mathbf{s}}_n]$ be the Gram-Schmidt orthogonalization of  $\mathbf{S}$. We refer to $\|\widetilde{\mathbf{S}}\|$ as the Gram-Schmidt norm of $\mathbf{S}$. For any $\mathbf{c}\in \mathbb{R}^n$ and  $\sigma \in \mathbb{R}^{+}$, define the following:
$\rho_{\sigma,\mathbf{c}}(\mathbf{x})=\exp(-\pi\frac{\|\mathbf{x}-\mathbf{c}\|^2}{\sigma^2})$ and $\rho_{\sigma,\mathbf{c}}(\Lambda)=\sum_{\mathbf{x}\in\Lambda}\rho_{\sigma,\mathbf{c}}(\mathbf{x})$ for any $\mathbf{x}\in\Lambda$.
Define the discrete Gaussian distribution  over the lattice $\Lambda$ with parameter $\sigma$ and center $\mathbf{c}$ to be
$D_{\Lambda,\sigma,\mathbf{c}}(\mathbf{x})={\rho_{\sigma,\mathbf{c}}(\mathbf{x})}/{\rho_{\sigma,\mathbf{c}}(\Lambda)}$ for any $\mathbf{x}\in \Lambda$. We often omit $\mathbf{c}$ if it is $\mathbf{0}$.

Let $n,m,q \in \mathbb{Z}^{+}$ with $q\geq 2$. For $\mathbf{A}\in \mathbb{Z}_q^{n\times m}$ and $\mathbf{u}\in \mathbb{Z}_q^n$ that admits a solution to the equation $\mathbf{A}\cdot \mathbf{x}=\mathbf{u}\bmod q$, define
$$\Lambda^{\bot}(\mathbf{A})=\{\mathbf{e}\in\mathbb{Z}^m: \mathbf{A}\mathbf{e}=\mathbf{0} \mod q\}, ~~\Lambda^{\mathbf{u}}(\mathbf{A})=\{\mathbf{e}\in\mathbb{Z}^m: \mathbf{A}\mathbf{e}=\mathbf{u} \mod q\}.$$
Define discrete Gaussian distribution over  the set $\Lambda^{\mathbf{u}}(\mathbf{A})$ in the following way:  $D_{\Lambda^{\mathbf{u}}(\mathbf{A}),\sigma,\mathbf{c}}(\mathbf{x})={\rho_{\sigma,\mathbf{c}}(\mathbf{x})}/{\rho_{\sigma,\mathbf{c}}(\Lambda^{\mathbf{u}}(\mathbf{A}))}$ for $\mathbf{x}\in\Lambda^{\mathbf{u}}(\mathbf{A})$.



\begin{lemma}[\cite{GPV08,PR06}]\label{lemma:infinity-norm-bound}
Let $n,m,q\in \mathbb{Z}^+$ with $q\geq 2$ and   $m\geq 2n\log q$. Let $\sigma\in \mathbb{R}$ such that $\sigma\geq \omega(\sqrt{\log m})$.
\begin{itemize}
\item Then for all but a $2q^{-n}$ fraction of all $\mathbf{A}\in\mathbb{Z}_q^{n\times m}$, the distribution of the syndrome $\mathbf{u}=\mathbf{A}\cdot\mathbf{e}\mod q$ is within negligible statistical distance from uniform over $\mathbb{Z}_q^n$ for $\mathbf{e}\hookleftarrow D_{\mathbb{Z}^m,\sigma}$. Besides, given $\mathbf{A}\cdot\mathbf{e}=\mathbf{u}\mod q$, the conditional distribution of $\mathbf{e}\hookleftarrow D_{\mathbb{Z}^m,\sigma}$  is $D_{\Lambda^{\mathbf{u}}(\mathbf{A}),\sigma}$.
\item Let ${x}\hookleftarrow D_{\mathbb{Z},\sigma}$, $t=\log n$, and $\beta=\lceil\sigma\cdot t\rceil$. Then the probability of $|{x}|\leq\beta$ is overwhelming.
\item The  distribution $D_{\mathbb{Z}^m,\sigma}$  has min-entropy at least $m-1$.
\end{itemize}

\end{lemma}
We next present two hard average-case  problems: the \emph{Short Integer Solution} ($\mathsf{SIS}$) problem (in the $\ell_\infty$ norm) and the \emph{Learning With Errors} ($\mathsf{LWE}$) problem.
\begin{definition}[\cite{Ajtai96,MR07,GPV08}, $\mathsf{SIS}_{n,m,q,\beta}^\infty$]Given $\mathbf{A}\xleftarrow{\$}\mathbb{Z}_q^{n\times m}$, find a vector $\mathbf{e}\in\mathbb{Z}^m$ so that $\mathbf{A}\cdot\mathbf{e}=\mathbf{0}\bmod q$ and $0<\|\mathbf{e}\|_{\infty}\leq \beta$.
\end{definition}
Let  $q> \beta \sqrt{n}$ be an integer and $m,\beta$ be polynomials in $n$, then solving the $\mathsf{SIS}_{n,m,q,\beta}^{\infty}$ problem (in the $\ell_{\infty}$ norm) is no easier than solving the  $\mathsf{SIVP}_{\gamma}$ problem in the worst-case for some $\gamma=\beta\cdot \widetilde{\mathcal{O}}(\sqrt{nm})$ (see \cite{GPV08,MP13}).

\begin{definition}[\cite{Regev05}, $\mathsf{LWE}_{n,q,\chi}$] 
For $\mathbf{s}\in\mathbb{Z}_q^n$, define a distribution $\mathcal{A}_{\mathbf{s},\chi} $ over $\mathbb{Z}_q^n\times \mathbb{Z}_q$ as follows: it samples a uniform vector $\mathbf{a}$ over $\mathbb{Z}_q^n$ and an element $e$ according to $\chi$,  and outputs the pair $(\mathbf{a},\mathbf{a}^\top\cdot \mathbf{s}+{e})$. Then the goal of the $\mathsf{LWE}_{n,q,\chi}$ problem is to distinguish $m=\mathrm{poly}(n)$ samples chosen according to the distribution $\mathcal{A}_{\mathbf{s},\chi}$ for some secret $\mathbf{s}\in\mathbb{Z}_q^n$ from $m$ samples chosen according to the uniform  distribution over $\mathbb{Z}_q^n\times \mathbb{Z}_q$.
\end{definition}
Let $B=\widetilde{\mathcal{O}}(\sqrt{n})$ and  $\chi$ be an efficiently samplable distribution over $\mathbb{Z}$  that outputs samples $e\in\mathbb{Z}$ with $|e|\leq B$ with all but negligible probability in $n$. If $q\geq 2$ is an arbitrary modulus, then the $\mathsf{LWE}_{n,q,\chi}$ problem is at least as hard as the worst-case problem $\mathsf{SIVP}_{\gamma}$ with $\gamma=\widetilde{\mathcal{O}}(n\cdot q/B)$ 
through an efficient quantum reduction~\cite{Regev05,BGV12}. Additionally, it is showed that the hardness of the $\mathsf{LWE}$ problem is  maintained when the secret $\mathbf{s}$ is chosen from the error distribution $\chi$ (see~\cite{ACPS09}).


Now let us recall some algorithms from previous works that will be used extensively in this work. 
\begin{lemma}[\cite{AP09}]\label{lemma:trapgen}
Let $n,m,q\in \mathbb{Z}^+$ with $q\geq 2$ and $m=O(n\log q)$.
There is a $\mathrm{PPT}$ algorithm $\mathsf{TrapGen}(n,m,q)$ which returns  a tuple $(\mathbf{A},\mathbf{S})$ such that \begin{itemize}
\item $\mathbf{A}$ is   within negligible statistical distance from uniform over $\mathbb{Z}_q^{n\times m}$,
\item $\mathbf{S}$ is a basis for $\Lambda^{\bot}(\mathbf{A})$, i.e., $\mathbf{A}\cdot\mathbf{S}=0\bmod q$, and $\|\widetilde{\mathbf{S}}\|\leq \mathcal{O}(\sqrt{n\log q})$.
\end{itemize}
\end{lemma}

\begin{lemma}[\cite{GPV08}]\label{lemma:sampled}
Let $\mathbf{S}\in\mathbb{Z}^{m\times m}$ be a basis of $\Lambda^{\bot}(\mathbf{A})$ for some $\mathbf{A}\in\mathbb{Z}_q^{n\times m}$ whose columns expand the entire group $\mathbb{Z}_q^n$. Let $\mathbf{u}$ be a  vector over $\mathbb{Z}_q^n$ and $s\geq \omega(\sqrt{\log n})\cdot\|\widetilde{\mathbf{S}}\|$. There is a $\mathrm{PPT}$ algorithm $\mathsf{SampleD}(\mathbf{A},\mathbf{S},\mathbf{u},s)$ which returns a vector $\mathbf{v}\in \Lambda^{\mathbf{u}}(\mathbf{A})$  from a distribution that is within negligible statistical distance from $D_{\Lambda^{\mathbf{u}}(\mathbf{A}),s}$.
\end{lemma}
We also need the following two algorithms to securely delegate basis.

\begin{lemma}[\cite{CHKP10}]\label{lemma:extbasis}
Let $\mathbf{S}\in\mathbb{Z}^{m\times m}$ be a basis of $\Lambda^{\bot}(\mathbf{A})$ for some $\mathbf{A}\in\mathbb{Z}_q^{n\times m}$ whose columns generate the entire group $\mathbb{Z}_q^n$. Let $\mathbf{A'}\in\mathbb{Z}_q^{n\times {m'}}$ be any matrix containing $\mathbf{A}$ as a submatrix. There is a deterministic polynomial-time algorithm $\mathsf{ExtBasis}(\mathbf{S},\mathbf{A}')$ which returns a basis $\mathbf{S'}\in\mathbb{Z}^{m'\times m'}$ of $\Lambda^{\bot}(\mathbf{A}')$ with $\|\widetilde{\mathbf{S}'}\|=\|\widetilde{\mathbf{S}}\|$.
\end{lemma}

\begin{lemma}[\cite{CHKP10}]\label{lemma:randbasis}
Let $\mathbf{S}$ be a basis of an $m$-dimensional integer lattice $\Lambda$ and a parameter $ s\geq \omega(\sqrt{\log n})\cdot\|\widetilde{\mathbf{S}}\|$. There is a $\mathrm{PPT}$ algorithm $\mathsf{RandBasis}(\mathbf{S},s)$ that outputs a new basis $\mathbf{S}'$ of $\Lambda$ with $\|\mathbf{S}'\|\leq s\cdot\sqrt{m}$.  Moreover, for any two bases  $\mathbf{S}_0,\mathbf{S}_1$ of $\Lambda$ and any $s\geq\mathrm{max}\{\|\widetilde{\mathbf{S}_0}\|,\|\widetilde{\mathbf{S}_1}\|\}\cdot\omega(\sqrt{\log n})$, the outputs of $\mathsf{RandBasis}(\mathbf{S}_0,s)$ and $\mathsf{RandBasis}(\mathbf{S}_1,s)$ are statistically close.

\end{lemma}


\subsection{The Bonsai Tree Signature Scheme}\label{subsection:bonsai-tree}
Our construction builds on the Bonsai tree signature scheme~\cite{CHKP10}. Now we describe it briefly. The scheme takes the following parameters: $\lambda$ is the security parameter and $n=\mathcal{O}(\lambda)$, $\ell$ is the message length, integer $q=\mathrm{poly}(n)$ is sufficiently large, $m=\mathcal{O}(n\log q)$, $\widetilde{L}=\mathcal{O}(\sqrt{n\log q})$, $s=\omega({\sqrt{\log n}})\cdot\widetilde{L}$, and $\beta=\lceil s\cdot\log n\rceil$. The verification key is the tuple $(\mathbf{A}_0,\mathbf{A}_1^0,\mathbf{A}_1^1,\ldots,\mathbf{A}_\ell^0,\mathbf{A}_\ell^1,\mathbf{u})$  while the signing key is $\mathbf{S}_0$, where $(\mathbf{A}_0,\mathbf{S}_0)$ is generated by the $\mathsf{TrapGen}(n,m,q)$ algorithm as described in Lemma \ref{lemma:trapgen} and  matrices $\mathbf{A}_1^0,\mathbf{A}_1^1,\ldots,\mathbf{A}_\ell^0,\mathbf{A}_\ell^1$ and vector $\mathbf{u}$ are all uniformly random and independent over $\mathbb{Z}_q^{n\times m}$ and $\mathbb{Z}_q^{n}$, respectively.

To sign a binary message $\mathrm{id}\in\{0,1\}^\ell$, the signer first computes the matrix $\mathbf{A}_{\mathrm{id}}:=[\mathbf{A}_0|\mathbf{A}_1^{\mathrm{id}[1]}|\cdots|\mathbf{A}_\ell^{\mathrm{id}[\ell]}]\in\mathbb{Z}_q^{n\times (\ell+1)m}$, and then outputs a vector $\mathbf{v}\in\Lambda^{\mathbf{u}}(\mathbf{A}_{\mathrm{id}})$ via the algorithm $\mathsf{SampleD}(\mathsf{ExtBasis}(\mathbf{S}_0,\mathbf{A}_{\mathrm{id}}),\mathbf{u},s)$. To verify the validity of $\mathbf{v}$ on message $\mathrm{id}$, the verifier  computes   $\mathbf{A}_{\mathrm{id}}$ as above and checks if  $\mathbf{A}_{\mathrm{id}}\cdot\mathbf{v}=\mathbf{u}\mod q$ and $ \|\mathbf{v}\|_{\infty}\leq \beta$ hold. They proved that this signature scheme is existential unforgeable under static chosen message attacks   based on the hardness of the $\mathsf{SIS}$ problem.

\subsection{Stern-Like Zero-Knowledge Argument Systems}\label{subsection:Stern}
The statistical zero-knowledge argument of knowledge ($\mathsf{ZKAoK}$) presented in this work are Stern-like~\cite{Ste96} protocols. In 1996, Stern~\cite{Ste96} suggested a three-move zero-knowledge protocol for the well-known syndrome decoding ($\mathsf{SD}$) problem. It was then later adapted to the lattice setting for a restricted version of Inhomogeneous Short Integer Solution ($\mathsf{ISIS}^{\infty}$) problem by Kawachi et al.~\cite{KTX08}. More recently, Ling et al.~\cite{LNSW13} generalized the protocol to handle more versatile relations that find applications in the designs of various lattice-based constructions (see, e.g.,~\cite{LLNW16,LLMNW16-ge,NguyenTW17,LLNW17-AC-PRF,LLMNW17-AC-OT,LLNW18}). Libert et al.~\cite{LLMNW16-dgs} put forward an abstraction of Stern's protocol to capture a wider range of lattice-based relations, which  we now recall.

\smallskip
\noindent
{\bf An abstraction of Stern's Protocol.} Let  $K, L, q\in \mathbb{Z}^+$ with $L\geq K$ and $q \geq 2$, and let $\mathsf{VALID}\subset\{-1,0,1\}^L$. Given a finite set $\mathcal{S}$, associate every $\phi \in \mathcal{S}$ with a permutation $\Gamma_\phi$ of $L$ elements so that the following conditions hold:
\begin{eqnarray}\label{eq:zk-equivalence}
\begin{cases}
\mathbf{w} \in \mathsf{VALID} \hspace*{2.5pt} \Longleftrightarrow \hspace*{2.5pt} \Gamma_\phi(\mathbf{w}) \in \mathsf{VALID}, \\
\text{If } \mathbf{w} \in \mathsf{VALID} \text{ and } \phi \text{ is uniform in } \mathcal{S}, \text{ then }  \Gamma_\phi(\mathbf{w}) \text{ is uniform in } \mathsf{VALID}.
\end{cases}
\end{eqnarray}
 The target is to construct a statistical \textsf{ZKAoK} for the abstract relation of the following form:
\begin{eqnarray*}
\mathrm{R_{abstract}} = \big\{(\mathbf{M}, \mathbf{u}), \mathbf{w} \in \mathbb{Z}_q^{K \times L} \times \mathbb{Z}_q^K \times \mathsf{VALID}: \mathbf{M}\cdot \mathbf{w} = \mathbf{u} \bmod q.\big\}
\end{eqnarray*}

To obtain the desired $\mathsf{ZKAoK}$ protocol, one has to prove that $\mathbf{w} \in \mathsf{VALID}$ and $\mathbf{w}$ satisfies the linear equation $\mathbf{M}\cdot \mathbf{w}=\mathbf{u}\bmod q$. To prove the former condition holds in a $\mathsf{ZK}$ manner, the prover chooses $\phi \xleftarrow{\$}\mathcal{S}$ and let the verifier check $\Gamma_\phi(\mathbf{w}) \in \mathsf{VALID}$. According to the first condition in~(\ref{eq:zk-equivalence}), the verifier should be convinced that $\mathbf{w}$ is indeed from the set $\mathsf{VALID}$. At the same time, the verifier is not able to learn any extra information about $\mathbf{w}$ due to the second condition in~(\ref{eq:zk-equivalence}). To show in $\mathsf{ZK}$ that the linear equation holds, the prover simply chooses $\mathbf{r}_w \xleftarrow{\$} \mathbb{Z}_q^L$ as a masking vector and then shows to the verifier that the equation $\mathbf{M}\cdot (\mathbf{w} + \mathbf{r}_w) = \mathbf{M}\cdot \mathbf{r}_w + \mathbf{u} \bmod q$ holds instead.


Figure~\ref{Figure:Interactive-Protocol} describes  the interaction between the prover $\mathcal{P}$ and the verifier $\mathcal{V}$ in details. The system utilizes  a statistically hiding and computationally binding string commitment scheme \textsf{COM}  from~\cite{KTX08}.

\begin{figure}[!htbp]

\begin{enumerate}
  \item \textbf{Commitment:} Prover chooses $\mathbf{r}_w \xleftarrow{\$} \mathbb{Z}_q^L$, $\phi \xleftarrow{\$} \mathcal{S}$ and randomness $\rho_1, \rho_2, \rho_3$ for $\mathsf{COM}$.
Then he sends $\mathrm{CMT}= \big(C_1, C_2, C_3\big)$ to the verifier, where
    \begin{gather*}
        C_1 =  \mathsf{COM}(\phi, \mathbf{M}\cdot \mathbf{r}_w \bmod q; \rho_1), \hspace*{5pt}
        C_2 =  \mathsf{COM}(\Gamma_{\phi}(\mathbf{r}_w); \rho_2), \\
        C_3 =  \mathsf{COM}(\Gamma_{\phi}(\mathbf{w} + \mathbf{r}_w \bmod q); \rho_3).
    \end{gather*}

  \item \textbf{Challenge:} $\mathcal{V}$ randomly choose a challenge $Ch$ from the set $ \{1,2,3\}$ and sends it to $\mathcal{P}$.
  \item \textbf{Response:} According to the choice of $Ch$, $\mathcal{P}$ sends back response $\mathrm{RSP}$ computed in the following manner:
  \smallskip
\begin{itemize}
\item $Ch = 1$: Let $\mathbf{t}_{w} = \Gamma_{\phi}(\mathbf{w})$, $\mathbf{t}_{r} = \Gamma_{\phi}(\mathbf{r}_w)$, and $\mathrm{RSP} = (\mathbf{t}_w, \mathbf{t}_r, \rho_2, \rho_3)$. 

\item $Ch = 2$: Let $\phi_2 = \phi$, $\mathbf{w}_2 = \mathbf{w} + \mathbf{r}_w \bmod q$, and
    $\mathrm{RSP} = (\phi_2, \mathbf{w}_2, \rho_1, \rho_3)$. 
\item $Ch = 3$: Let $\phi_3 = \phi$, $\mathbf{w}_3 = \mathbf{r}_w$, and
 $\mathrm{RSP} = (\phi_3, \mathbf{w}_3, \rho_1, \rho_2)$.
\end{itemize}
\end{enumerate}
\textbf{Verification:}  When receiving $\mathrm{RSP}$ from $\mathcal{P}$, $\mathcal{V}$ performs as follows:
\smallskip
          \begin{itemize}
            \item $Ch = 1$: Check that $\mathbf{t}_w \in \mathsf{VALID}$, $C_2 = \mathsf{COM}(\mathbf{t}_r; \rho_2)$, ${C}_3 = \mathsf{COM}(\mathbf{t}_w + \mathbf{t}_r \bmod q; \rho_3)$. \smallskip

             \item $Ch = 2$: Check that $C_1 = \mathsf{COM}(\phi_2, \mathbf{M}\cdot \mathbf{w}_2 - \mathbf{u} \bmod q; \rho_1)$, ${C}_3 = \mathsf{COM}(\Gamma_{\phi_2}(\mathbf{w}_2); \rho_3)$. \smallskip

            \item $Ch = 3$: Check that $C_1 =  \mathsf{COM}(\phi_3, \mathbf{M}\cdot \mathbf{w}_3; \rho_1), \hspace*{5pt}
        C_2 =  \mathsf{COM}(\Gamma_{\phi_3}(\mathbf{w}_3); \rho_2).$

          \end{itemize}
          In each case, $\mathcal{V}$ returns $1$ if and only if all the conditions hold.
\caption{Stern-like \textsf{ZKAoK} for the relation $\mathrm{R_{abstract}}$.}
\label{Figure:Interactive-Protocol}
\end{figure}

\begin{theorem}[\cite{LLMNW16-dgs}]\label{theorem:zk-protocol}
Let $\mathsf{COM}$ be a statistically hiding and computationally binding string commitment scheme. Then the interactive depicted in Figure~\ref{Figure:Interactive-Protocol} is a statistical \emph{\textsf{ZKAoK}} with perfect completeness, soundness error~$2/3$, and communication cost~$\mathcal{O}(L\log q)$. Specifically:
\begin{itemize}
\item There exists a polynomial-time simulator that, on input $(\mathbf{M}, \mathbf{u})$, with probability $2/3$ it outputs an accepted transcript that is within negligible statistical distance from the one produced by an honest prover who knows the witness.
\item There is an algorithm that, takes as inputs $(\mathbf{M},\mathbf{u})$ and three accepting transcripts  $(\mathrm{CMT},1,\mathrm{RSP}_1)$, $(\mathrm{CMT},2,\mathrm{RSP}_2)$, $(\mathrm{CMT},3,\mathrm{RSP}_3)$ on $(\mathbf{M},\mathbf{u})$, and outputs $\mathbf{w}' \in \mathsf{VALID}$ such that $\mathbf{M}\cdot \mathbf{w}' = \mathbf{u} \bmod q$ in  polynomial time.
\end{itemize}
\end{theorem}
The proof of the Theorem~\ref{theorem:zk-protocol}, appeared in~\cite{LLMNW16-dgs}, is omitted here.

\section{Our Lattice-Based Forward-Secure Group Signature}\label{section:main-scheme}

In the description below, for a binary tree of depth $k$, we identify each node at depth~$j$ with a binary vector $z$ of length $j$ such that $z[1]$ to $z[j]$ are ordered from the top to the bottom and a $0$ and a $1$ indicate the left and right branch respectively in the order of traversal. Let $B\in \mathbb{Z}^+$.  For an integer $0\leq b\leq B$, denote $\mathsf{Bin}(b)$ as the binary representation of $b$ with length $\lceil\log B \rceil$.

In our FSGS scheme, lifetime of the scheme is divided into $T$ discrete periods $ 0,1,\cdots,T-1 $.
For simplicity, let $T=2^d$ for some  $d\in \mathbb{Z}^+$.
Following previous works~\cite{BSSW06,LY10},  each time period $t$ is associated with leaf $\mathsf{Bin}(t)$.

Following \cite{BSSW06}, for $j=1,\cdots,d+1$, $t\in\{0,1,\cdots, T-1\}$,  we define a time period's ``right sibling at depth $j$'' as
\[
 \mathrm{sibling}(j,t) =
  \begin{cases}
  (1)^\top &~\text{if }~ j=1~\text{and}~ \mathsf{Bin}(t)[j]=0, \\
   (\mathsf{Bin}(t)[1],\cdots,\mathsf{Bin}(t)[j-1],1)^\top& ~\text{if }~ 1< j\leq d ~\text{and}~ \mathsf{Bin}(t)[j]=0, \\
    \bot &~ \text{if }~ j\leq d ~\text{and}~ \mathsf{Bin}(t)[j]=1, \\
     \mathsf{Bin}(t) & ~\text{if }~j=d+1.
  \end{cases}
\]
Define node set $\mathrm{Nodes}_{(t\rightarrow T-1)}$ to be $\{\mathrm{sibling}(1,t),\cdots,\mathrm{sibling}(d+1,t)\}$. 
For any $t'>t$, one can check that for any non-$\bot$ $z'\in\mathrm{Nodes}_{(t'\rightarrow T-1)}$, there exists a $z\in\mathrm{Nodes}_{(t\rightarrow T-1)}$ such that $z$ is an ancestor of $z'$.

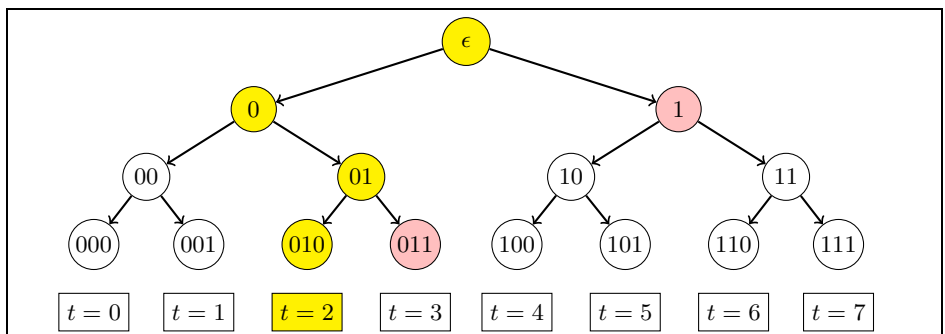
\begin{figure}
  \centering
  \begin{tikzpicture}[scale=0.9]
  \node[circle,draw, inner sep=4.5pt,fill=yellow] (root) at  (0,0) {$\epsilon$};

  \node[circle,draw, inner sep=1.5pt] (000) at (-5.5,-3) {$000$};
  \node[circle,draw, inner sep=1.5pt] (111) at (5.5,-3) {$111$};

  \node[circle,draw, inner sep=1.5pt,fill=pink] (011) at (-0.75,-3) {$011$};
  \node[circle,draw, inner sep=1.5pt] (100) at (0.75,-3) {$100$};

  \node[circle,draw, inner sep=1.5pt,fill=yellow] (010) at (-2.35,-3) {$010$};
  \node[circle,draw, inner sep=1.5pt] (101) at (2.35,-3) {$101$};

  \node[circle,draw, inner sep=1.5pt] (001) at (-3.95,-3) {$001$};
  \node[circle,draw, inner sep=1.5pt] (110) at (3.95,-3) {$110$};

  \node[rectangle,draw, minimum size=0.5cm] (x000) at (-5.5,-4) {$t=0$};
  \node[rectangle,draw, minimum size=0.5cm] (x111) at (5.5,-4) {$t=7$};
  \node[rectangle,draw, minimum size=0.5cm] (x011) at (-0.75,-4) {$t=3$};
  \node[rectangle,draw, minimum size=0.5cm] (x100) at (0.75,-4) {$t=4$};

  \node[rectangle,draw, minimum size=0.5cm, fill=yellow] (x010) at (-2.35,-4) {$t=2$};
  \node[rectangle,draw, minimum size=0.5cm] (x101) at (2.35,-4) {$t=5$};

  \node[rectangle,draw, minimum size=0.5cm] (x001) at (-3.95,-4) {$t=1$};
  \node[rectangle,draw, minimum size=0.5cm] (x110) at (3.95,-4) {$t=6$};

  \node[circle,draw, inner sep=2.5pt] (00) at (-4.725,-2) {$00$};
  \node[circle,draw, inner sep=2.5pt] (11) at (4.725,-2) {$11$};

  \node[circle,draw, inner sep=2.5pt,fill=yellow] (01) at (-1.55,-2) {$01$};
  \node[circle,draw, inner sep=2.5pt] (10) at (1.55,-2) {$10$};

  \node[circle,draw, inner sep=3.5pt,fill=yellow] (0) at (-3.1375,-1) {$0$};
  \node[circle,draw, inner sep=3.5pt, fill=pink] (1) at (3.1375,-1) {$1$};

  \draw [<-, thick] (0) edge (root);
  \draw [<-, thick] (1) edge (root);

  \draw [<-, thick] (00) edge (0);
  \draw [<-, thick] (01) edge (0);
  \draw [<-, thick] (10) edge (1);
  \draw [<-, thick] (11) edge (1);

  \draw [<-, thick] (000) edge (00);
  \draw [<-, thick] (001) edge (00);
  \draw [<-, thick] (010) edge (01);
  \draw [<-, thick] (011) edge (01);
  \draw [<-, thick] (100) edge (10);
  \draw [<-, thick] (101) edge (10);
  \draw [<-, thick] (110) edge (11);
  \draw [<-, thick] (111) edge (11);

  \end{tikzpicture}
  \caption{A binary tree with time periods $T=2^3$. Consider the path from the root $\epsilon$ to the leaf node $(010)^\top$, we have $\mathrm{sibling}(1,2)=(1)^\top$, $\mathrm{sibling}(2,2)=\bot$, $\mathrm{sibling}(3,2)=(011)^\top$, and $\mathrm{sibling}(4,2)=(010)^\top$.  Therefore,  $\mathrm{Nodes}_{(2\rightarrow 7)}=\{(1)^\top, \bot, (011)^\top, (010)^\top\}$.  Similarly, $\mathrm{Nodes}_{(5\rightarrow 7)}=\{ \bot, (11)^\top,\bot, (101)^\top\}$. It is verifiable that  node $(1)^\top$ is an ancestor of both node $(101)^\top$ and node $(11)^\top$.   }\label{tree-illustration}
\end{figure}

\subsection{Description of the Scheme}\label{subsection:construction}
Our scheme operates in the Nakanishi et al.'s (static) model~\cite{NHF09}. Let $T=2^d$ and $N=2^\ell$. The group public key consists of (i) a Bonsai tree of depth $\ell+d$ specified by a matrix $\mathbf{A}=[\mathbf{A}_0|\mathbf{A}_1^0|\mathbf{A}_1^1\cdots|\mathbf{A}_{\ell+d}^0|\mathbf{A}_{\ell+d}^1]\in\mathbb{Z}_q^{n\times(2\ell+2d+1)m}$ and a vector $\mathbf{u}\in\mathbb{Z}_q^n$, which are  for issuing certificate; (ii) A public matrix $\mathbf{B}\in\mathbb{Z}_q^{n\times m}$ of the $\mathrm{IBE}$ scheme by Gentry et al.~\cite{GPV08}, which is for encrypting user identities when signing messages. The secret key of  $\mathsf{GM}$ is a trapdoor matrix of the Bonsai tree while the secret key of the tracing manager is a  trapdoor matrix of the $\mathrm{IBE}$ scheme.

Each user $\mathrm{id}\in\{0,1\}^{\ell}$ is assigned a node $\mathrm{id}$. To enable periodical key updating, each user $\mathrm{id}$ is associated with a subtree of depth $d$. In our scheme, all users are assumed to be valid group members from time $0$ to $T-1$. Let $z$ be a binary string of length $d_z\leq d$. Define $\mathbf{A}_{\mathrm{id}\|z}=[\mathbf{A}_0|\mathbf{A}_1^{\mathrm{id}[1]}|\cdots|\mathbf{A}_\ell^{\mathrm{id}[\ell]}|\mathbf{A}_{\ell+1}^{z[1]}|\cdots|\mathbf{A}_{\ell+d_z}^{z[d_z]}]\in\mathbb{Z}_q^{n\times (\ell+d_z+1)m}$. Specifically, the group signing key of user $\mathrm{id}$ at time $t$ is  $\{\mathbf{S}_{\mathrm{id}\|z},z\in\mathrm{Nodes}_{(t\rightarrow T-1)}\}$, which satisfies $\mathbf{A}_{\mathrm{id}\|z}\cdot \mathbf{S}_{\mathrm{id}\|z}=\mathbf{0}\bmod q$. Thanks to the basis delegation technique~\cite{CHKP10}, users are able to compute the trapdoor matrices for all the descendent of nodes in the set $\mathrm{Nodes}_{(t\rightarrow T-1)}$ and hence are able to derive all the subsequent signing keys.We remark that for leaf nodes, it is sufficient to  generate   short vectors   instead of short bases, since we do not need to perform further delegations.

Once received the group signing key, each user can issue signatures on behalf of the group. When signing a message at time $t$, user $\mathrm{id}$ first generates a one-time signature key pair $(\mathsf{ovk},\mathsf{osk})$, and then encrypts his identity  $\mathrm{id}$  to a ciphertext $\mathbf{c}$ using the $\mathrm{IBE}$ scheme with respect to ``identity'' $\mathsf{ovk}$. Next, he proves in zero-knowledge that: (i) he is a certified group member; (ii) he has done key evolution honestly; (iii) $\mathbf{c}$ is a correct encryption of $\mathrm{id}$. To prove that facts (i) and (ii) hold, it is sufficient to prove knowledge of a short vector $\mathbf{v}_{\mathrm{id}\|\mathsf{Bin}(t)}$  such that $\mathbf{A}_{\mathrm{id}\|\mathsf{Bin}(t)}\cdot \mathbf{v}_{\mathrm{id}\|\mathsf{Bin}(t)}=\mathbf{u}\bmod q$. The protocol is developed from Langlois et al.'s technique~\cite{LLNW14-PKC}(which was also employed in~\cite{ChengNW16} for designing policy-based signatures) and Ling et al.'s technique~\cite{LNSW13}, and is repeated $\kappa=\omega(\log n)$ times to achieve negligible soundness error, and is made non-interactive via Fiat-Shamir transform~\cite{FS86} as a triple $\Pi$. Finally, the user generates a one-time signature $\mathrm{sig}$ on the pair $(\mathbf{c},\Pi)$, and outputs the group signature consisting of $(\mathsf{ovk},\mathbf{c},\Pi,\mathrm{sig})$.

To verify a group signature, one  checks the validity of $\mathrm{sig}$ under the key $\mathsf{ovk}$ and $\Pi$. In case of dispute, $\mathsf{TM}$ can decrypt the ciphertext with respect to the ``identity'' $\mathsf{ovk}$ using his secret key.
Details of the scheme are described below.
\begin{description}
\item $\mathsf{KeyGen}(\lambda,T,N)$: On inputs security parameter $\lambda$, total number of time periods $T=2^d$ for some $d\in\mathbb{Z}_{+}$ and maximum number of group members $N=2^\ell$ for some $\ell\in\mathbb{Z}^{+}$, this algorithm does the following:
    \smallskip
        \begin{enumerate}
        \item Choose $n=\mathcal{O}(\lambda)$,   $q=\mathrm{poly}(n)$,  $m=\mathcal{O}(n\log q)$. Let $k=\ell+d$ and $\kappa=\omega(\log n)$.
        \item Run $\mathsf{TrapGen}(n,m,q)$ as described in Lemma~\ref{lemma:trapgen} to obtain $\mathbf{A}_0\in\mathbb{Z}_q^{n\times m}$ and $\mathbf{S}_0\in\mathbb{Z}^{m\times m}$.

        \item Sample $\mathbf{u}\xleftarrow{\$}\mathbb{Z}_q^n$, and $\mathbf{A}_{i}^{b}\xleftarrow{\$}\mathbb{Z}_q^{n\times m}$ for all $i\in[k]$ and $b\in\{0,1\}$. 

        \item Choose a one-time signature scheme $\mathcal{OTS}=(\mathsf{OGen}, \mathsf{OSign},\mathsf{OVer})$, and  a statistically hiding and computationally binding commitment scheme  $\mathsf{COM}$ from~\cite{KTX08}    that will be used in our zero-knowledge argument system.
       \item Let $\mathcal{H}_0: \{0,1\}^*\rightarrow \mathbb{Z}_q^{n\times \ell}$ and $\mathcal{H}_{1}:\{0,1\}^*\rightarrow \{1,2,3\}^{\kappa}$   be collision-resistant hash functions, which will be modelled as random oracles in the security analysis.
      \item Let Gaussian parameter $s_i$ be $\mathcal{O}(\sqrt{nk\log q})^{i-\ell+1}\cdot\omega(\sqrt{\log n})^{i-\ell+1}$, which will be used to generate short bases or sample short vectors at level $i$ for $i\in\{\ell,\ell+1,\cdots,k\}$. 

       \item Choose integer bounds $\beta=\lceil s_{k}\cdot \log n\rceil, B=\widetilde{\mathcal{O}}(\sqrt{n})$, and let $\chi$ be a $B$-bounded distribution over $\mathbb{Z}$.


        \item Generate a master key pair $(\mathbf{B},\mathbf{S})\in\mathbb{Z}_q^{n\times m}\times \mathbb{Z}^{m\times m}$ for the  IBE scheme by Gentry et al.~\cite{GPV08} via the $\mathsf{TrapGen}(n,m,q)$ algorithm. \smallskip

        \item
             For   user $i\in\{0,1,\cdots,N-1\}$, let  $\mathrm{id}=\mathsf{Bin}(i)\in\{0,1\}^\ell$.  Let node $\mathrm{id}$ be the identifier of user $i$.   Determine the node set $\mathrm{Nodes}_{(0\rightarrow T-1)}$.

             For  $z\in\mathrm{Nodes}_{(0\rightarrow T-1)}$, if $z=\bot$, set $\mathbf{usk}_0[i][z]=\bot$. Otherwise denote $d_z$ as the length of $z$ with $d_z\leq d$,  compute the matrix \[\mathbf{A}_{\mathrm{id}\|z}=[\mathbf{A}_0|\mathbf{A}_1^{\mathrm{id}[1]}|\cdots|\mathbf{A}_\ell^{\mathrm{id}[\ell]}|\mathbf{A}_{\ell+1}^{z[1]}|\cdots|\mathbf{A}_{\ell+d_z}^{z[d_z]}]\in\mathbb{Z}_q^{n\times (\ell+d_z+1)m}.\]
            and proceed as follows:
                \begin{itemize}

                \item If $z$ is of length $d$, i.e., $d_z=d$, it computes a vector $\mathbf{v}_{\mathrm{id}\|z}\in \Lambda^{\mathbf{u}}(\mathbf{A}_{\mathrm{id}\|z})$ via $$\mathbf{v}_{\mathrm{id}\|z}\leftarrow\mathsf{SampleD}(\mathsf{ExtBasis}(\mathbf{S}_0,\mathbf{A}_{\mathrm{id}\|z}),\mathbf{u},s_{k}).$$  Set $\mathbf{usk}_{0}[i][z]=\mathbf{v}_{\mathrm{id}\|z}$.  \smallskip

                \item If $z$ is of length less than $d$, i.e., $1\leq d_z<d$, it  computes a matrix $\mathbf{S}_{\mathrm{id}\|z}\in\mathbb{Z}^{(\ell+d_z+1)m\times(\ell+d_z+1)m}$ via
                    \[\mathbf{S}_{\mathrm{id}\|z}\leftarrow\mathsf{RandBasis}(\mathsf{ExtBasis}(\mathbf{S}_0,\mathbf{A}_{\mathrm{id}\|z}),s_{\ell+d_z}).\]
                   Set $\mathbf{usk}_{0}[i][z]=\mathbf{S}_{\mathrm{id}\|z}$.
                \smallskip
                \end{itemize}
            Let  $\mathbf{usk}_{0}[i]=\{\mathbf{usk}_{0}[i][z],z\in\mathrm{Nodes}_{(0\rightarrow T-1)}\}$ be the initial  secret key of user~$i$.
        \end{enumerate}
Let public parameter be $\mathsf{pp}$, group public key be $\mathsf{gpk}$, secret key of $\mathsf{GM}$ be $\mathsf{msk}$, secret key of $\mathsf{TM}$ be $\mathsf{mosk}$ and initial secret key be $\mathbf{usk}_0$, which are defined as follows:
\[\mathsf{pp}=\{n,q,m,\ell,d,k,\kappa,\mathcal{OTS},\mathsf{COM},\mathcal{H}_0,\mathcal{H}_{1},s_\ell,\ldots,s_k,\beta,B\},\]
\[\mathsf{gpk}=\{\mathsf{pp},\mathbf{A}_0,\mathbf{A}_{1}^{0},\mathbf{A}_{1}^{1},\ldots,\mathbf{A}_{k}^{0},\mathbf{A}_{k}^{1},\mathbf{u},\mathbf{B}\},\]
\[\mathsf{msk}=\mathbf{S}_0,~~~~~\mathsf{mosk}=\mathbf{S}, \]
\[\mathbf{usk}_0=\{\mathbf{usk}_0[0],\hspace{2pt}\mathbf{usk}_0[1],\hspace{2pt}\ldots,\hspace{2pt}\mathbf{usk}_0[N-1] \}.\]


\smallskip

\item $\mathsf{KeyUpdate}(\mathsf{gpk},\mathbf{usk}_{t}[i],i,t+1)$: Compute the identifier of user $i$ as  $\mathrm{id}=\mathsf{Bin}(i)$,  parse $\mathbf{usk}_{t}[i]=\{\mathbf{usk}_{t}[i][z],z\in\mathrm{Nodes}_{(t\rightarrow T-1)}\}$, and determine the node set $\mathrm{Nodes}_{(t+1\rightarrow T-1)}$. \smallskip

For $z'\in\mathrm{Nodes}_{(t+1\rightarrow T-1)}$, if $z'=\bot$, set $\mathbf{usk}_{t+1}[i][z']=\bot$. Otherwise,
there exists exactly one $z\in\mathrm{Nodes}_{(t\rightarrow T-1)}$ as its prefix, i.e., $z'=z\|y$ for some suffix $y$. 
Consider the following two cases. \smallskip
        \begin{enumerate}
        \item If  $z'=z$, let $\mathbf{usk}_{t+1}[i][z']=\mathbf{usk}_{t}[i][z]$.
        \smallskip
        \item If $z'=z\|y$ for some non-empty $y$,  then $\mathbf{usk}_{t}[i][z]=\mathbf{S}_{\mathrm{id}\|z}$. Consider the following two subcases.
            \begin{itemize}
            \item If $z'$ is of length $d$,  run \[\mathbf{v}_{\mathrm{id}\|z'}\leftarrow\mathsf{SampleD}(\mathsf{ExtBasis}(\mathbf{S}_{\mathrm{id}\|z},\mathbf{A}_{\mathrm{id}\|z'}),\mathbf{u},s_{k}),\]
            and set $\mathbf{usk}_{t+1}[i][z']=\mathbf{v}_{\mathrm{id}\|z'}$.  \smallskip

            \item  If $z'$ is of length less than $d$, run \[\mathbf{S}_{\mathrm{id}\|z'}\leftarrow\mathsf{RandBasis}(\mathsf{ExtBasis}(\mathbf{S}_{\mathrm{id}\|z},\mathbf{A}_{\mathrm{id}\|z'}),s_{\ell+d_{z'}}),\]
            and set $\mathbf{usk}_{t+1}[i][z']=\mathbf{S}_{\mathrm{id}\|z'}$. \smallskip
            \end{itemize}
        \end{enumerate}
Output updated key as $\mathbf{usk}_{t+1}[i]=\{\mathbf{usk}_{t+1}[i][z'],z'\in\mathrm{Nodes}_{(t+1\rightarrow T-1)}\}$.

\smallskip
\item $\mathsf{Sign}(\mathsf{gpk},\mathbf{usk}_{t}[i],i,t,M)$: 
Compute the identifier  $\mathrm{id}=\mathsf{Bin}(i)$. 
By the structure of the node set $\mathrm{Nodes}_{(t\rightarrow T-1)}$, there exists some $z\in\mathrm{Nodes}_{(t\rightarrow T-1)}$ such that $z=\mathsf{Bin}(t)$ is of length $d$  and $\mathbf{usk}_{t}[i][z]=\mathbf{v}_{\mathrm{id}\|z}$. \smallskip

To sign a message $M\in\{0,1\}^*$, the signer then performs the following steps.\smallskip

        \begin{enumerate}

            \item \label{sign:step 1} First, generate a one-time signature key pair $(\mathsf{ovk}, \ms{osk}) \leftarrow \mathsf{OGen}(n)$,
              and then encrypt  $\mathrm{id}$ with respect to ``identity'' $\mathsf{ovk}$ as follows.
                            Let $\mathbf{G} =
              \mathcal{H}_0(\ms{ovk})\in\mathbb{Z}_q^{n\times \ell}$. Sample $\mathbf{s} \hookleftarrow \chi^n$, $\mathbf{e}_1\hookleftarrow \chi^m$, $\mathbf{e}_2 \hookleftarrow\chi^\ell$, and  compute
               ciphertext $(\mathbf{c}_1,\mathbf{c}_2)\in\mathbb{Z}_q^{m}\times\mbb{Z}_q^\ell$  as
               \begin{equation}\label{eq:ciphertext-of-id}
               (\mathbf{c}_1=\mathbf{B}^\top\cdot\mathbf{s}+\mathbf{e}_1, \hspace{6.8pt} \mathbf{c}_2=\mathbf{G}^\top\cdot\mathbf{s}+\mathbf{e}_2+\lfloor\frac{q}{2}\rfloor\cdot\mathrm{id}).
               \end{equation}
            \item\label{sign:step 2}  Second, compute the matrix $\mathbf{A}_{\mathrm{id}\|z}$ and generate a $\mathsf{NIZKAoK}$ $\Pi$ to demonstrate the possession of a valid tuple
                \begin{equation}\xi=(\mathrm{id},\mathbf{s},\mathbf{e}_1,\mathbf{e}_2,\mathbf{v}_{\mathrm{id}\|z})\label{witness tuple}\end{equation}
                 such that
                    \begin{enumerate}
                    \item \label{witness condition 1}$\mathbf{A}_{\mathrm{id}\|z}\cdot \mathbf{v}_{\mathrm{id}\|z}=\mathbf{u}\mod q$,
                    and $\|\mathbf{v}_{\mathrm{id}\|z}\|_{\infty}\leq \beta$.
                    \smallskip
                    \item \label{witness condition 2} 
                    Equations in~(\ref{eq:ciphertext-of-id}) hold with $\|\mathbf{s}\|_{\infty}\leq B$, $\|\mathbf{e}_1\|_{\infty}\leq B$ and $\|\mathbf{e}_2\|_{\infty}\leq B$.
                   \smallskip
                    \end{enumerate}
        This is done by running our argument system described  in Section \ref{subsection:main-nizk}
        with public input
        \[\gamma=(\mathbf{A}_0,\mathbf{A}_{1}^{0},\mathbf{A}_{1}^{1},\ldots,\mathbf{A}_{k}^{0},\mathbf{A}_{k}^{1},\mathbf{u},\mathbf{B},\mathbf{G},\mathbf{c}_1,\mathbf{c}_2,t)\]
        and  witness tuple $\xi$ as above. 
        The protocol is repeated $\kappa=\omega(\log n)$ times to obtain negligible soundness error and made non-interactive via the Fiat-Shamir heuristic~\cite{FS86} as a triple $\Pi=((\mathrm{CMT}_i)_{i=1}^{\kappa},\mathrm{CH},(\mathrm{RSP}_i)_{i=1}^{\kappa})$ with $\mathrm{CH}=\mathcal{H}_{1}(M,(\mathrm{CMT}_i)_{i=1}^{\kappa},\mathbf{c}_1,\mathbf{c}_2,t).$
        \item \label{sign:step 3}Third, compute a one-time signature $\mathrm{sig} = \ms{OSign}(\ms{osk}; \mathbf{c}_1, \mathbf{c}_2, \Pi)$ and output the signature as $\Sigma= (\ms{ovk},\mathbf{c}_1, \mathbf{c}_2, \Pi,\mathrm{sig})$.\smallskip
      \end{enumerate}

\item $\mathsf{Verify}(\ms{gpk},t,M,\Sigma)$: This algorithm proceeds as follows:
     \begin{enumerate}
        \item Parse  $\Sigma$ as $\Sigma=(\ms{ovk},\mathbf{c}_1, \mathbf{c}_2,\Pi,\mathrm{sig})$.  If
         $\ms{OVer}(\ms{ovk}; \mathrm{sig}; \mathbf{c}_1, \mathbf{c}_2, \Pi) = 0$, then return $0$.
         \smallskip

        \item Parse $\Pi$ as $\Pi=((\mathrm{CMT}_i)_{i=1}^{\kappa},(\mathrm{Ch}_1,\ldots,\mathrm{Ch}_\kappa),(\mathrm{RSP}_i)_{i=1}^{\kappa})$.  

            If $(\mathrm{Ch}_1,\cdots,\mathrm{Ch}_\kappa)\neq\mathcal{H}_{1}(M,(\mathrm{CMT}_i)_{i=1}^{\kappa},\mathbf{c}_1,\mathbf{c}_2,t)$, then return $0$.

        \smallskip
        \item  For $i\in[\kappa]$, run
the verification step of the underlying argument protocol  
to check the validity of $\mathrm{RSP}_i$ with respect to
$\mathrm{CMT}_i$ and $\mathrm{Ch}_i$. If any of the conditions does not hold, then return $0$.
\item Return $1$.

     \end{enumerate}

\item $\mathsf{Open}(\ms{gpk},\ms{mosk},t,M,\Sigma)$: If $\mathsf{Verify}(\ms{gpk},t,M,\Sigma)=0$, abort. Otherwise,
let $\ms{mosk}$ be $\mathbf{S}\in\mathbb{Z}^{m\times m}$ and parse $\Sigma$ as $\Sigma=(\ms{ovk},\mathbf{c}_1, \mathbf{c}_2,\Pi, \mathrm{sig})$. Then it decrypts $(\mathbf{c}_1,\mathbf{c}_2)$ as follows:
    \begin{enumerate}
    \item Compute $\mathbf{G}=\mathcal{H}_0(\mathsf{ovk})=[\mathbf{g}_1|\cdots|\mathbf{g}_\ell]\in\mathbb{Z}^{n\times \ell}$.  Then use $\mathbf{S}$ to compute a small norm matrix $\mathbf{F}_{\ms{ovk}}\in\mathbb{Z}^{m\times \ell}$ such that $\mathbf{B}\cdot \mathbf{F}_{\mathsf{ovk}}=\mathbf{G}\bmod q$. This is done by computing $\mathbf{f}_i\leftarrow \mathsf{SampleD}(\mathbf{B},\mathbf{S},\mathbf{g}_i,s_\ell)$ for all $i\in[\ell]$ and let $\mathbf{F}_{\ms{ovk}}=[\mathbf{f}_1|\cdots|\mathbf{f}_{\ell}]$.
        \smallskip
    \item Use $\mathbf{F}_{\ms{ovk}}$ to decrypt $(\mathbf{c}_1,\mathbf{c}_2)$ by computing $$\mathrm{id'}=\big\lfloor\frac{\mathbf{c}_2-\mathbf{F}_{\ms{ovk}}^\top\cdot \mathbf{c}_1}{\lfloor q/2\rfloor}\big\rceil\in\{0,1\}^\ell.$$
    \item Return $\mathrm{id}'\in\{0,1\}^\ell$.
    \end{enumerate}

\end{description}

\subsection{Analysis of the Scheme}

\noindent
{\sc Efficiency. }
We first analyze the complexity of the scheme described in Section~\ref{subsection:construction}, with respect to security parameter $\lambda$ and parameters $\ell =  \log N $ and $d= \log T $. Recall $k=\ell+d$.
\begin{itemize}
\item The group public key $\mathsf{gpk}$   has bit-size $\widetilde{\mathcal{O}}(\lambda^2\cdot k)$.

\item The user secret key $\mathbf{usk}_t[i]$ has at most $d+1$ trapdoor matrices, and has bit-size $\widetilde{\mc{O}}(\lambda^2\cdot k^2d)$.

\item The size of signature $\Sigma$ is dominated by that of the Stern-like $\mathsf{NIZKAoK}$ $\Pi$, which is $\widetilde{\mathcal{O}}(|\xi|\cdot \log q)\cdot \omega(\log \lambda)$, where $|\xi|$ denotes the bit-size of the witness-tuple $\xi$. Overall, $\Sigma$ has bit-size $\widetilde{\mathcal{O}}(\lambda\cdot k)$.
\end{itemize}

\noindent
{\sc Correctness. }
The correctness of the above scheme follows from the following facts: (i) the underlying argument system is perfectly complete; (ii) the underlying encryption scheme obtained by transforming the IBE scheme in~\cite{GPV08} via CHK transformation~\cite{CHK04} is correct.

Specifically, for an honest user, when he signs a message at time period $t$, he is able to demonstrate the possession of a valid tuple $\xi$ of the form~(\ref{witness tuple}).  Therefore, with probability $1$, the resulting signature $\Pi$ will be accepted by the $\mathsf{Verify}$ algorithm, implied by the perfect completeness of the underlying argument system.  As for the correctness of the $\ms{Open}$ algorithm, note that
\begin{align*}
\mathbf{c}_2-\mathbf{F}_{\ms{ovk}}^\top\cdot \mathbf{c}_1 &= \mathbf{G}^\top\cdot\mathbf{s}+\mathbf{e}_2+\big\lfloor\frac{q}{2}\big\rfloor\cdot\mathrm{id}-\mathbf{F}_{\ms{ovk}}^\top\cdot(\mathbf{B}^\top\cdot\mathbf{s}+\mathbf{e}_1) \\
&=\big\lfloor\frac{q}{2}\big\rfloor\cdot\mathrm{id}+\mathbf{e}_2-\mathbf{F}_{\ms{ovk}}^\top\cdot\mathbf{e}_1
\end{align*}
where $\|\mathbf{e}_1\|\leq B$, $\|\mathbf{e}_2\|_{\infty}\leq B$, and $\|\mathbf{f}_{i}\|_{\infty}\leq \lceil  s_\ell\cdot \log m\rceil=\widetilde{\mathcal{O}}( \sqrt{n\cdot k})$, which is implied by Lemma~\ref{lemma:infinity-norm-bound}.  Recall that  $q=\mathrm{poly}(n)$, $m=\mathcal{O}(n\log q)$ and $B=\widetilde{\mathcal{O}}(\sqrt{n})$. Hence
\[\|\mathbf{e}_2-\mathbf{F}_{\ms{ovk}}^\top\cdot\mathbf{e}_1\|_{\infty} 
\leq B+m\cdot B\cdot\widetilde{\mathcal{O}}( \sqrt{n\cdot k})=\widetilde{\mathcal{O}}(n^2).\]
As long as we choose sufficiently large $q$, with probability $1$, the $\mathsf{Open}$ algorithm will recover $\mathrm{id}$ and correctness of the $\mathsf{Open}$ algorithm holds.

\smallskip

\noindent
{\sc Security. } In Theorem~\ref{theorem:main-scheme}, we prove that our scheme satisfies the security requirements put forward by Nakanishi et al.~\cite{NHF09}.
\begin{theorem}\label{theorem:main-scheme}
In the random oracle model, the forward-secure group signature described
in Section \ref{subsection:construction} satisfies  full anonymity and forward-secure traceability requirements under the $\mathsf{LWE}$ and $\mathsf{SIS}$ assumptions.
\end{theorem}

The proof of the theorem is established by Lemma \ref{lemma:anonymity} and Lemma \ref{lemma:traceability}.

\begin{lemma}\label{lemma:anonymity}
Suppose that one-time signature scheme $\mathcal{OTS}$ is strongly unforgeable. In the random oracle model, the forward-secure group signature scheme described
in Section \ref{subsection:construction} is fully anonymous under the hardness of the $\mathsf{LWE}_{n,q,\chi}$ problem.
\end{lemma}
\begin{proof}

Denote $\mc{C}$ as the challenger and $\mathcal{A}$ as the adversary. Following~\cite{LNW15}, we prove this lemma using a series of computationally indistinguishable games. The first game Game $0$ is the real experiment $\mathbf{Exp}_{\mathrm{FSGS},\mathcal{A}}^{\ms{Anon}}(\lambda,T,N)$ while  the last game is such that the advantage of the adversary is~$0$.
\begin{description}

\item[Game $0$:] In this game,  $\mathcal{C}$ runs the experiment $\mathbf{Exp}_{\mathrm{FSGS},\mathcal{A}}^{\mathsf{Anon}}(\lambda,T,N)$ faithfully.
In the challenge phase,  $\mathcal{A}$ outputs a  message $M^*$ together with two  users $0\leq i_0,i_1\leq N-1$ for the targeted time $t^*$.   $\mathcal{C}$ responds by sending back a signature $\Sigma^*= (\ms{ovk}^*,\mathbf{c}_1^*, \mathbf{c}_2^*, \Pi^*, \mathrm{sig}^*)\leftarrow\mathsf{Sign}(\mathsf{gpk},\mathbf{usk}_{t^*}[i_b],t^*,M^*)$ for a random bit $b\in\{0,1\}$. 
Then the adversary outputs a bit $b'\in\{0,1\}$ and this game returns~$1$ if $b'=b$ and $0$ otherwise.
In this experiment, $\mathcal{C}$ replies with all random strings for oracles queries of $\mathcal{H}_0,\mathcal{H}_1$.  

\smallskip
\item[Game $1$:] In this game, we modify Game $0$ in two aspects: (i) We generate the pair $(\ms{ovk}^*,\ms{osk}^*)$ in the very beginning of the experiment; (ii) For the signature opening queries, if  $\mathcal{A}$ asks for a valid signature of the form $\Sigma= (\ms{ovk},\mathbf{c}_1, \mathbf{c}_2, \Pi, \mathrm{sig})$ such that $\ms{ovk}=\ms{ovk}^*$, then  $\mathcal{C}$ outputs a random bit and aborts the experiment. Now we argue that the probability that $\mathcal{C}$ aborts is negligible and hence Game $0$ and Game $1$ are computationally indistinguishable. Actually, before the challenged signature is given to  $\mathcal{A}$, $\mathsf{ovk}^*$ is independent of $\mc{A}'s$ view, hence it is negligible that $\mathcal{A}$ queries a signature containing $\ms{ovk}^*$. Furthermore, after the challenged signature is sent to $\mathcal{A}$, if $\mathcal{A}$ queries a \emph{new} valid signature of the form $(\ms{ovk}^*,\mathbf{c}_1,\mathbf{c}_2,\Pi, \mathrm{sig})$, then $((\mathbf{c}_1,\mathbf{c}_2,\Pi),\mathrm{sig})$ is a successful forgery of the $\mathcal{OTS}$ scheme, which breaks the strong unforgeability of the $\mathcal{OTS}$ scheme. This proves that $\mathcal{C}$ aborts with negligible probability. From now on, we  assume that $\mathcal{A}$ will not query valid signature containing $\mathsf{ovk}^*$.

\smallskip

\item[Game $2$:] In this game, we change Game $1$ in the following ways. First, instead of generating $\mathbf{B}$  using the $\mathsf{TrapGen}$ algorithm, we generate  a uniformly random matrix $\mathbf{B}^*$ over $\mathbb{Z}_q^{n\times m}$. This change is indistinguishable to $\mathcal{A}$ since the matrix $\mathbf{B}$ is statistically close to uniform by Lemma~\ref{lemma:trapgen}.
    Second, we program the random oracle $\mathcal{H}_0$ as follows. For query of $\mathsf{ovk}^*$, it generates a uniformly random matrix $\mathbf{G}^*\in\mathbb{Z}_q^{n\times \ell}$, let $\mathcal{H}_0(\ms{ovk}^*)=\mathbf{G}^*$, and return $\mathbf{G}^*$ to $\mathcal{A}$. This change also makes no difference to $\mathcal{A}$ since the output of $\mathcal{H}_0$ is uniformly random. For query of $\mathsf{ovk}\neq\mathsf{ovk}^*$, we first sample $\mathbf{F}_{\ms{ovk}}\leftarrow D_{\mathbb{Z},s_\ell}^{m\times \ell}$, let $\mathcal{H}_0(\mathsf{ovk})=\mathbf{G}=\mathbf{B}^*\cdot\mathbf{F}_{\ms{ovk}}\mod q$, and return $\mathbf{G}$ to $\mathcal{A}$. We then keep a record of $(\mathsf{ovk},\mathbf{F}_{\ms{ovk}},\mathbf{G})$. $\mathbf{G}$ generated in this way is statistically close to uniform by Lemma~\ref{lemma:infinity-norm-bound}. Hence, this change does not affect $\mathcal{A}$'s view non-negligibly.
    When $\mathcal{A}$  queries the opening oracle a signature $(\ms{ovk},\mathbf{c}_1,\mathbf{c}_2,\Pi, \mathrm{sig})$ with $\mathsf{ovk}\neq \ms{ovk}^*$, $\mathcal{C}$ can use the recorded~$\mathbf{F}_{\ms{ovk}}$ to decrypt the ciphertext $(\mathbf{c},\mathbf{c}_2)$. 
    Hence $\mathcal{C}$ can answer all signature opening queries. It then follows that Game~$2$ and Game~$1$ are statistically indistinguishable.

\smallskip

\item[Game $3$:] In this game,we modify  Game $2$ as follows.
Instead of generating a real proof $\Pi^*$ for the challenged signature,  we generate a simulated one without using the witness by programming the random oracle $\mathcal{H}_1$.
Since our argument system is statistically zero-knowledge, the view of adversary $\mathcal{A}$ is statistically indistinguishable between Game~$3$ and Game $2$. 

\smallskip

\item[Game $4$:]  In this game, we change Game $3$ as follows. Instead of computing  \[(\mathbf{c}_1^*, \mathbf{c}_2^*)=(\mathbf{B^*}^\top\cdot\mathbf{s}+\mathbf{e}_1, ~ \mathbf{G^*}^\top\cdot\mathbf{s}+\mathbf{e}_2+\big\lfloor\frac{q}{2}\big\rfloor\cdot\mathsf{Bin}(i_b))\in\mathbb{Z}_q^m\times \mathbb{Z}_q^{\ell}\] in the challenged phase, where   $\mathbf{s}\in\chi^n,\mathbf{e}_1\in\chi^m,\mathbf{e}_2\in\chi^\ell$, we let $$(\mathbf{c}_1^*, \mathbf{c}_2^*)=(\mathbf{z}_1^*,~\mathbf{z}_2^*+\big\lfloor\frac{q}{2}\big\rfloor\cdot\mathsf{Bin}(i_b) ),$$ where $\mathbf{z}_1^*,\mathbf{z}_2^*$ are uniformly random vectors over $\mathbb{Z}_q^m$ and $\mathbb{Z}_q^\ell$. We claim that this modification is computationally indistinguishable to the view of the adversary $\mathcal{A}$ assuming the hardness of $\mathsf{LWE}_{n,q,\chi}$. Indeed, if we let $\mathbf{D}=[\mathbf{B}^*|\mathbf{G}^*]\in\mathbb{Z}_q^{n\times (m+\ell)}$, $\mathbf{e}=(\mathbf{e}_1\|\mathbf{e}_2)\in\chi^{m+\ell}$, and $\mathbf{z}=(\mathbf{z}_1^*\|\mathbf{z}_2^*)\in\mathbb{Z}_q^{m+\ell}$, then to distinguish Game $3$ and Game $4$ is to distinguish $(\mathbf{D}, \mathbf{D}^T\cdot\mathbf{s}+\mathbf{e})$ and $(\mathbf{D},\mathbf{z})$. Recall that $\mathbf{B}^*$ and $\mathbf{G}^*$ are uniformly random matrices, so is $\mathbf{D}$.  

\smallskip

\item[Game $5$:]  In this game, we slightly change Game $4$ by substituting $(\mathbf{c}_1^*, \mathbf{c}_2^*)$ with a new independent and uniform tuple $(\mathbf{z}'_1,\mathbf{z}'_2)$. 
    It is straightforward  that Game $5$ and Game $4$ are statistically indistinguishable. 
    Furthermore, the challenged signature in this game does not depends on the challenged bit $ b$ any more, and hence the advantage of $\mathcal{A}$ in this game is $0$.
\smallskip
\end{description}

 It then follows that $\mathbf{Adv}_{\mathrm{FSGS},\mathcal{A}}^{\ms{Anon}}(\lambda,T,N)$ is negligible in $\lambda$ because of the indistinguishability of the above games. This concludes the proof. \qed 

\end{proof}

\begin{lemma}
\label{lemma:traceability} In the random oracle model, the   forward-secure group signature scheme described in Section \ref{subsection:construction} is
 forward-secure traceable under the hardness of the $\mathsf{SIS}_{n,\overline{m},q,2\beta}^{\infty}$ problem, where $\overline{m}=(k+1)m$. 
\end{lemma}
\begin{proof}

Assume there is a $\mathrm{PPT}$ adversary $\mathcal{A}$ attacking the forward-secure traceability of the forward-secure group signature scheme with non-negligible probability, then we construct a new $\mathrm{PPT}$ adversary  $\mathcal{B}$ attacking the $\mathsf{SIS}_{n,\overline{m},q,2\beta}^{\infty}$ problem with non-negligible probability.

Given an $\mathsf{SIS}$ instance $\mathbf{C}\in\mathbb{Z}_q^{n\times\overline{m}}$, the goal of $\mathcal{B}$ is to find a non-zero vector $\mathbf{v}\in\mathbb{Z}_q^{\overline{m}}$ such that $\mathbf{C}\cdot\mathbf{v}=\mathbf{0}\mod q$ and $\|\mathbf{v}\|_{\infty}\leq 2\beta$. 
Towards this goal, $\mathcal{B}$ simulates the view of the adversary $\mathcal{A}$ attacking the forward-secure traceability.
Initially, $\mathcal{B}$ lets $t=0$ and the set $\mathrm{CU}$ be empty and then generates group public key $\mathsf{gpk}$, secret key for opening $\mathsf{mosk}$, and some secret internal state as follows:
    \begin{itemize}
        \item Parse $\mathbf{C}$ as $\mathbf{C}=[\mathbf{C}_0|\mathbf{C}_1|\cdots|\mathbf{C}_{k}]$ for $\mathbf{C}_j\in\mathbb{Z}_q^{n\times m}$, $j\in\{0,1,\cdots,k\}$. It then generates the remaining public parameters as in Section~\ref{subsection:construction}.
        \item Sample $\mathbf{z}=(\mathbf{z}_0\|\mathbf{z}_1\|\cdots\|\mathbf{z}_{k})\in\mathbb{Z}^{\overline{m}}$, where each $\mathbf{z}_i$ is sampled from $D_{\mathbb{Z}^{m},s_{k}}$. If $\|\mathbf{z}\|_{\infty}>\beta$, then repeat the sampling. Compute $\mathbf{u}=\mathbf{C}\cdot\mathbf{z}\mod q$. 

        \item Guess the targeted  user $i^*\in\{0,1,\cdots,N-1\}$ and targeted forgery time period $t^*\in\{0,1,\cdots,T-1\}$ uniformly. 

        \item \emph{Uncontrolled growth.} Let $\mathrm{id}^*=\mathsf{Bin}(i^*)$ and $z^*=\mathsf{Bin}(t^*)$. Define $\mathbf{A}_{j}^{\mathrm{id}^*[j]}$ to be $\mathbf{C}_{j}$ for $j\in[\ell]$ and $\mathbf{A}_{\ell+j}^{z^*[j]}$ to be $\mathbf{C}_{\ell+j}$ for $j\in[d]$. Recall $k=\ell+d$. 

        \item \emph{Controlled growth.} Generate $\mathbf{A}_{j}^{1-\mathrm{id}^*[j]}$ for all $j\in[\ell]$ via $(\mathbf{A}_{j}^{1-\mathrm{id}^*[j]},\mathbf{S}_j)\leftarrow\mathsf{TrapGen}(n,m,q)$ and $\mathbf{A}_{\ell+j}^{1-z^*[j]}$  for $j\in[d]$ via the algorithm $(\mathbf{A}_{\ell+j}^{1-z^*[j]},\mathbf{S}_{\ell+j})\leftarrow\mathsf{TrapGen}(n,m,q)$. 

        \item Generate a master key pair $(\mathbf{B},\mathbf{S})$ via $\mathsf{TrapGen}(n,m,q)$. 

    \end{itemize}
$\mathcal{B}$ invokes $\mathcal{A}$ by sending  $\mathsf{gpk}, \mathsf{mosk}$ and then interacts with $\mathcal{A}$.   At the start of each time period $t\in\{0,1,\cdots, T-1\}$, $\mathcal{B}$ announces the beginning of $t$ to $\mathcal{A}$. At current time period $t$, $\mathcal{B}$ responds to $\mathcal{A}$'s queries as follows.
        \begin{itemize}
         \item  When $\mathcal{A}$ queries the random oracles $\mathcal{H}_0,\mathcal{H}_1$, $\mathcal{B}$ replies   with uniformly random strings and keeps a record of the queries. 

         \item When $\mathcal{A}$ queries the secret key of member $i^*$, if $i^*\in \mathrm{CU}$ or $t\leq t^*$,  $\mathcal{B}$ then aborts\footnote{ This guarantees that once the secret key of member $i^*$ is queried, it is only queried at time $t$ such that $ t>t^*$.}. Otherwise, for each node $z\in\mathrm{Nodes}_{(t\rightarrow T-1)}$,  it first computes the smallest index $d_{z,t}$ such that $1\leq d_{z,t}\leq d$ and $ z^* [d_{z,t}]\neq z[d_{z,t}]$.
        Then $\mathcal{B}$ computes $\mathbf{usk}_{t}[i][z]$ via $\mathsf{SampleD}(\mathsf{ExtBasis}(\mathbf{S}_{\ell+d_{z,t}},\mathbf{A}_{\mathrm{id}^*\|z}),s_{k})$ if $z$ is of length $d$ or via $\mathsf{RandBasis}(\mathsf{ExtBasis}(\mathbf{S}_{\ell+d_{z,t}},\mathbf{A}_{\mathrm{id}^*\|z}),s_{\ell+d_{z}})$ if $z$ is of length $d_z<d$.
        Finally, $\mathcal{B}$ sets $\mathbf{usk}_t[i]$  as in our construction and sends it to $\mathcal{A}$. Add $i^*$ to the set $\mathrm{CU}$. 

         \item When $\mathcal{A}$ queries the secret key of  member $i\neq i^*$, if $i\in \mathrm{CU}$, $\mathcal{B}$ then aborts. Otherwise,  let $\mathrm{id}=\mathsf{Bin}(i)$ and $\ell_i$ be the smallest index such that $1 \leq \ell_i\leq \ell$ and $\mathrm{id}[\ell_i]\neq \mathrm{id}^*[\ell_i]$.
         Compute $\mathbf{usk}_{t}[i][z]$  via the algorithm $\mathsf{SampleD}(\mathsf{ExtBasis}(\mathbf{S}_{\ell_i},\mathbf{A}_{\mathrm{id}\|z}),\mathbf{u},s_{k})$ if $z$ is of length $d$ or via the algorithm  $\mathsf{RandBasis}(\mathsf{ExtBasis}(\mathbf{S}_{\ell_i},\mathbf{A}_{\mathrm{id}\|z}),s_{\ell+d_z})$ if $z$ is of length $d_z< d$ for $z\in\mathrm{Nodes}_{(t\rightarrow T-1)}$.
         Set $\mathbf{usk}_t[i]$  as in our construction and send it to $\mathcal{A}$. Finally, add $i$ to the set $\mathrm{CU}$. \smallskip

        \item When $\mathcal{A}$ queries a signature on a message for user $i$. If $i\in \mathrm{CU}$ at current time $t$, $\mathcal{B}$ aborts\footnote{Note that $\mathcal{B}$ never answers the signing queries of corrupted users.}.  Otherwise, if $i\neq i^*$, $\mathcal{B}$ responds as in our algorithm $\mathsf{Sign}$ using the
        corresponding witness. If $i=i^*$, $\mathcal{B}$
         performs the same as in our algorithm $\mathsf{Sign}$ except that it generates a simulated proof $\Pi'$
         by programming the hash oracle $\mathcal{H}_1$ and returns
         $\Sigma= (\ms{ovk},\mathbf{c}_1, \mathbf{c}_2, \Pi',\mathrm{sig})$ to~$\mathcal{A}$.

        \end{itemize}
It is worth noticing that the secret key of user $i^*$ is either never queried ($i^*\notin \mathrm{CU}$) or is only queried at time $t$ with $t>t^*$.

We claim that $\mathcal{A}$ cannot distinguish whether it interacts with the real challenger or with $\mathcal{B}$. First, group public key $\mathsf{gpk}$ given to $\mathcal{A}$ is indistinguishable from the real one. This is because the output matrix $\mathbf{A}$ of the $\mathsf{TrapGen}$ algorithm is statistically close to uniform by Lemma~\ref{lemma:trapgen} and $\mathbf{u}$ is statistically close to a uniform vector over $\mathbb{Z}_q^n$ by Lemma~\ref{lemma:infinity-norm-bound}.
Second, the secret signing key given to   $\mathcal{A}$ is indistinguishable from the real one due to the fact that the outputs of $\mathsf{RandBasis}$ using two different bases are within negligible statistical distance by Lemma~\ref{lemma:randbasis}. Third, the signature queries make no difference to the view of $\mathcal{A}$. This can be implied by the statistical zero-knowledge property of the underlying argument system.

When $\mathcal{A}$ halts and outputs a message $M^*$ and a signature $\Sigma^*$ at the targeted time period $t'$ such that $\mathsf{Verify}(\ms{gpk},t',M^*,\Sigma^*)=1$ and $\Pi^*$ is not obtained by making a signing query at $M^*$, check $t'=t^*$ holds or not. If not, indicating the guess of $t^*$ fails, then $\mathcal{B}$ aborts. Otherwise, $\mathcal{A}$ outputs $(t^*,M^*,\Sigma^*)$. Parse $\Sigma^*$ as \[(\ms{ovk}^*,~\mathbf{c}_1^*, ~\mathbf{c}_2^*,~
(\mathrm{CMT}_i^*)_{i=1}^{\kappa},~(\mathrm{Ch}_i^*)_{i=1}^{\kappa},~(\mathrm{RSP}_i^*)_{i=1}^{\kappa},~\mathrm{sig}^*). \]
 Run $\mathrm{id}'\leftarrow\mathsf{Open}(\ms{gpk},\mathsf{mosk},t^*,M^*,\Sigma^*)$. If $\mathrm{id}'\neq\mathrm{id}^*$, indicating that the guess of $i^*$ fails, then $\mathcal{B}$ aborts. Otherwise, $\mathcal{B}$ makes use of the forgery to solve the $\mathsf{SIS}$ problem as follows.

 First,  $\mathcal{A}$ must have queried $\mathcal{H}_1$ for the tuple $(M^*,(\mathrm{CMT}_i^*)_{i=1}^{\kappa},\mathbf{c}_1^*,\mathbf{c}_2^*,t^*)$, 
 since the probability of guessing this value is at most $3^{-\kappa}$, which is negligible by our choice of $\kappa$.
 Let $(M^*,(\mathrm{CMT}_i^*)_{i=1}^{\kappa},\mathbf{c}_1^*,\mathbf{c}_2^*,t^*)$ be the $h$-th oracle query and $Q_{\mathcal{H}_1}$ be the total oracle queries $\mathcal{A}$ has made to $\mathcal{H}_1$.
Next, $\mathcal{B}$ lets $h$
be the targeted forking point and replays $\mathcal{A}$ polynomial-number times. For each new run, $\mathcal{B}$ starts with the same random tape  and random input  as in the original run. Further, for the first $h-1$ queries of $\mathcal{H}_1$, $\mathcal{B}$ replies with the same random value as in the original run, but from $h$-th query on, $\mathcal{B}$ replies with fresh and independent value. Besides, for queries of $\mathcal{H}_0$, $\mathcal{B}$ always replies as in the original run.

Constructed in this way, $(M^*,(\mathrm{CMT}_i^*)_{i=1}^{\kappa},\mathbf{c}_1^*,\mathbf{c}_2^*,t^*)$ is always the $h$-th oracle query $\mathcal{A}$ made to $\mathcal{H}_1$. The improved forking lemma~\cite{BPVY00} implies that with probability  $\geq\frac{1}{2}$, $\mathcal{B}$ obtains $3$-fork
for the  tuple $(M^*,(\mathrm{CMT}_i^*)_{i=1}^{\kappa},\mathbf{c}_1^*,\mathbf{c}_2^*,t^*)$ with pairwise
distinct hash values $\mathrm{CH}_h^{(1)}$, $\mathrm{CH}_h^{(2)}$, $\mathrm{CH}_h^{(3)}$ and corresponding valid responses $\mathrm{RSP}_h^{(1)}$, $\mathrm{RSP}_h^{(2)}$, $\mathrm{RSP}_h^{(3)}$. A simple calculation shows
that  with probability $1-(\frac{7}{9})^\kappa$, we have $\{\mathrm{CH}_{h,j}^{(1)}$, $\mathrm{CH}_{h,j}^{(2)}$, $\mathrm{CH}_{h,j}^{(3)}\}=\{1,2,3\}$ for some $j\in[\kappa]$. Therefore,
$\mathrm{RSP}_{h,j}^{(1)}$, $\mathrm{RSP}_{h,j}^{(2)}$, $\mathrm{RSP}_{h,j}^{(3)}$ are $3$ valid responses for all the challenges $1,2,3$   w.r.t. the same
commitment $\mathrm{CMT}_{j}^*$. Since $\mathsf{COM}$ is computationally binding, $\mathcal{B}$ is able to extract the witness tuple $$\xi^*=(\mathrm{id},\mathbf{s},\mathbf{e}_1,\mathbf{e}_2,\mathbf{v}_{\mathrm{id}\|z})$$ such that $\|\mathbf{v}_{\mathrm{id}\|z}\|_{\infty}\leq \beta$, $\|\mathbf{s}\|_{\infty}\leq B$, $\|\mathbf{e}_1\|_{\infty}\leq B$, $\|\mathbf{e}_2\|_{\infty}\leq B$ and
\begin{eqnarray*}
&\mathbf{A}_{\mathrm{id}\|z}\cdot \mathbf{v}_{\mathrm{id}\|z}=\mathbf{u}\mod q, \\
&(\mathbf{c}_1^*=\mathbf{B}^T\cdot\mathbf{s}+\mathbf{e}_1, ~ \mathbf{c}_2^*=(\mathbf{G}^*)^T\cdot\mathbf{s}+\mathbf{e}_2+\big\lfloor\frac{q}{2}\big\rfloor\cdot\mathrm{id})\in \mathbb{Z}_q^m\times \mathbb{Z}_q^{\ell},
\end{eqnarray*} where $\mathbf{G}^*=\mathcal{H}_0(\mathsf{ovk}^*)$.
Conditioned on guessing correctly $i^*,t^*$, we have $\mathrm{id}=\mathrm{id}^*$ and $z=z^*$. Therefore, $\mathbf{A}_{\mathrm{id}\|z}=\mathbf{C}$.
Now we have $\mathbf{C}\cdot \mathbf{v}_{\mathrm{id}\|z}=\mathbf{u}=\mathbf{C}\cdot\mathbf{z}\mod q$.
We claim that $\mathbf{v}_{\mathrm{id}\|z}\neq\mathbf{z}$ with overwhelming probability.
This is because $\mathcal{A}$ either queried the secret key of
user $\mathrm{id}$ at time after $t^*$ or never queried the secret key at all (by successfully attacking forward-secure traceability), then $\mathbf{z}$ is not known to $\mathcal{A}$.
Further, from the view of the adversary, $\mathbf{z}$ is from the distribution $D_{\Lambda^{\mathbf{u}}(\mathbf{C}),s_{k}}$ and hence has large min-entropy, which are implied by Lemma~\ref{lemma:infinity-norm-bound}. Therefore, $\mathbf{v}_{\mathrm{id}\|z}\neq\mathbf{z}$ with overwhelming probability. Hence $\mathbf{x}=\mathbf{v}_{\mathrm{id}\|z}-\mathbf{z}\neq \mathbf{0}$ and $\|\mathbf{x}\|_{\infty}\leq 2\beta$.  This implies that we solve the $\mathsf{SIS}_{n,\overline{m},q,2\beta}^{\infty}$ problem with non-negligible probability and hence our scheme is forward-secure traceable.

\end{proof}

\section{The Underlying Zero-Knowledge Argument System}\label{section:main-nizk}

In Section~\ref{subsection:dec-ext}, we recall the extension, decomposition, and permutation techniques from~\cite{LNSW13,LLNW14-PKC}. Then we describe in Section~\ref{subsection:main-nizk} our statistical $\mathsf{ZKAoK}$ protocol that will be used in generating group signatures.

\subsection{Extension, Decomposition, and Permutation}\label{subsection:dec-ext}
{\sc{Extensions.}} For $\mathfrak{m}\in\mathbb{Z}$, let $\mathsf{B}_{3\mathfrak{m}}$ be the set of all vectors in $\{-1,0,1\}^{3\mathfrak{m}}$ having exactly $\mathfrak{m}$ coordinates $-1$, $\mathfrak{m}$ coordinates $1$, and $\mathfrak{m}$ coordinates $0$ and $\mathcal{S}_{\mathfrak{m}}$ be the set of all permutations on $\mathfrak{m}$ elements. Let $\oplus$ be the addition operation modulo~$2$. Define the following functions
\begin{itemize}

\item $\mathsf{ext}_3$: $\{-1,0,1\}^{\mathfrak{m}}\rightarrow \mathsf{B}_{3\mathfrak{m}}$ that transforms a vector $\mathbf{v}=(v_1,\ldots,v_\mathfrak{m})^\top$  to vector $(\mathbf{v}\|(\mathbf{-1})^{\mathfrak{m}-n_{-1}}\|\mathbf{0}^{\mathfrak{m}-n_0}\|\mathbf{1}^{\mathfrak{m}-n_1})^\top$,  where $n_{j}$ is the number of element $j$ in the vector $\mathbf{v}$ for $j\in\{-1,0,1\}$.
\item $\mathsf{enc}_2$: $\{0,1\}^{\mathfrak{m}}\rightarrow \{0,1\}^{2\mathfrak{m}}$ that transforms a vector $\mathbf{v}=(v_1,\ldots,v_\mathfrak{m})^\top$ to vector $(v_1,1-v_1,\ldots,v_\mathfrak{m},1-v_\mathfrak{m})^\top$.
\end{itemize}

\noindent {\sc{Decompositions and Permutations.}} We now recall the integer decomposition technique. For any $B \in \mathbb{Z}^+$,  define   $p_B=\lfloor \log B\rfloor +1$ and  the sequence $B_1, \ldots, B_{p_B}$ as $B_j = \lfloor\frac{B + 2^{j-1}}{2^j} \rfloor$  for each $  j \in [p_B]$. As observed in \cite{LNSW13}, it  satisfies $\sum_{j=1}^{p_B} B_j = B$ and
 any integer $v \in [ B]$ can be decomposed to $\mathsf{idec}_B(v) = (v^{(1)}, \ldots, v^{(p_B)})^\top \in \{0,1\}^{p_B}$ such that $\sum_{j=1}^{p_B}B_j \cdot  v^{(j)} = v$. This decomposition procedure is described in a deterministic manner as follows: \vspace{-0.1 cm}
\begin{enumerate}
\item $v': = v$
\item For $j=1$ to $p_B$ do:
    \begin{enumerate}[(i)]
    \item If $v' \geq B_j$ then $v^{(j)}: = 1$, else $v^{(j)}: = 0$;
    \item $v': = v' - B_j\cdot v^{(j)}$.
    \end{enumerate}
\item Output $\mathsf{idec}_B(v) = (v^{(1)}, \ldots, v^{(p_B)})^\top$.
\end{enumerate}

\noindent Next, for any positive integers $\mathfrak{m}, B$, we define the function  $\mathsf{vdec}_{\mathfrak{m},B}$  that transforms 
  a vector $\mathbf{w}=(w_1,\ldots, w_\mathfrak{m})^\top\in [-B,B ]^{\mathfrak{m}}$ to a vector 
of the following form: $$\mathbf{w}'=(\sigma(w_1)\cdot \mathsf{idec}_{B}(|w_1|)\|\cdots\|\sigma(w_\mathfrak{m})\cdot \mathsf{idec}_B(|w_\mathfrak{m}|))\in\{-1,0,1\}^{\mathfrak{m}p_B},$$
where  $\forall j\in[\mathfrak{m}]$: $\sigma(w_j)=0$ if $w_j=0$;  $\sigma(w_j)=-1$ if $w_j<0$; $\sigma(w_j)=1$ if $w_j>0$.

Define the  matrix  $\mathbf{H}_{{\mathfrak{m}},B}=\left[
                                                                             \begin{array}{ccc}
                                                                               B_1,\ldots, B_{p_B}  & ~ & ~ \\
                                                                               ~ & \ddots & ~ \\
                                                                               ~ & ~& B_1,\ldots, B_{p_B}  \\
                                                                             \end{array}
                                                                           \right]\in\mathbb{Z}^{{\mathfrak{m}}\times { \mathfrak{m}}p_{B}}$ and its extension $\widehat{\mathbf{H}}_{\mathfrak{m},B}=[\mathbf{H}_{\mathfrak{m},B}|\mathbf{0}^{\mathfrak{m}\times 2\mathfrak{m}p_B}]\in\mathbb{Z}^{{\mathfrak{m}}\times { 3\mathfrak{m}}p_{B}}$.
Let $\widehat{\mb{w}}=\mathsf{ext}_3(\mathbf{w}')\in\mathsf{B}_{3\mathfrak{m}p_B}$, then one can see that $\widehat{\mathbf{H}}_{\mathfrak{m},B}\cdot \widehat{\mathbf{w}}=\mathbf{w}$ and for any $\psi\in \mathcal{S}_{3\mathfrak{m}p_B}$, the following equivalence holds:
 \begin{eqnarray}\label{eqn:permutation}
\widehat{\mathbf{w}}\in\mathsf{B}_{3\mathfrak{m}p_B}\Leftrightarrow \psi(\widehat{\mathbf{w}})\in\mathsf{B}_{3\mathfrak{m}p_B}.
\end{eqnarray}

Define the following permutation.
\begin{itemize}
\item For any $\mathbf{e}=(e_1,\ldots,e_{\mathfrak{m}})^\top \in\{0,1\}^{\mathfrak{m}}$, define   $\Pi_{\mathbf{e}}: \mathbb{Z}^{2\mathfrak{m}}\rightarrow \mathbb{Z}^{2\mathfrak{m}}$ that maps a vector $\mathbf{v}=(v_1^0,v_1^1,\ldots, v_{\mathfrak{m}}^0,v_{\mathfrak{m}}^1)^\top$ to $(v_1^{e_1},v_1^{1-e_1},\ldots, v_{\mathfrak{m}}^{e_\mathfrak{m}},v_{\mathfrak{m}}^{1-e_{\mathfrak{m}}})^\top$.


\end{itemize}
One can see that, for any $\mathbf{z},\mathbf{e}\in\{0,1\}^{\mathfrak{m}}$, the following equivalence  holds:
    \begin{eqnarray}\label{eqn:enc2}
    \mathbf{v}=\mathsf{enc}_2(\mathbf{z})\Leftrightarrow \Pi_{\mathbf{e}} (\mathbf{v})=\mathsf{enc}_2(\mathbf{z}\oplus \mathbf{e}).
    \end{eqnarray}

\subsection{The Underlying Zero-Knowledge Argument System}\label{subsection:main-nizk}
We now describe a statistical $\mathsf{ZKAoK}$ that will be invoked by the signer when generating group signatures. The protocol is developed from Stern-like techniques proposed by Ling et al.~\cite{LNSW13} and Langlois et al.~\cite{LLNW14-PKC}.
\begin{description}
\item {\bf{Public input}} $\gamma$: $\mathbf{A}_0\in\mathbb{Z}_q^{n\times m}$, $\mathbf{A}_{j}^{b}\in\mathbb{Z}_q^{n\times m}$ for $(b,j)\in\{0,1\}\times[k]$, $\mathbf{u}\in\mathbb{Z}_q^n$, $\mathbf{B}\in\mathbb{Z}_q^{n\times m}$, $\mathbf{G}\in\mathbb{Z}_q^{n\times\ell}$, $(\mathbf{c}_1,\mathbf{c}_2)\in\mathbb{Z}_q^m\times\mathbb{Z}_q^\ell$, $t\in\{0,1,\cdots,T-1\}$.

\smallskip

\item {\bf{Secret input}} $\xi$: $\mathrm{id}\in\{0,1\}^\ell$, $\mathbf{s}\in\chi^n$, $\mathbf{e}_1\in\chi^m$, $\mathbf{e}_2\in\chi^\ell$, $\mathbf{v}_{\mathrm{id}\|z}\in\mathbb{Z}^{(\ell+d+1)m}$ with $z=\mathsf{Bin}(t)$.

\smallskip

\item {\bf{Prover's goal}}:
\begin{eqnarray}\label{equation:goal-1}
\begin{cases}
\mathbf{A}_{\mathrm{id}\|z}\cdot \mathbf{v}_{\mathrm{id}\|z}=\mathbf{u}\bmod q,~\|\mathbf{v}_{\mathrm{id\|z}}\|_{\infty}\leq \beta;\\
                   \mathbf{c}_1=\mathbf{B}^\top\cdot\mathbf{s}+\mathbf{e}_1\bmod q, ~ \mathbf{c}_2=\mathbf{G}^\top\cdot\mathbf{s}+\mathbf{e}_2+\big\lfloor\frac{q}{2}\big\rfloor\cdot\mathrm{id}\bmod q;\\
                   \|\mathbf{s}\|_{\infty}\leq B,~ \|\mathbf{e}_1\|_{\infty}\leq B, ~\|\mathbf{e}_2\|_{\infty}\leq B.
\end{cases}
\end{eqnarray}

\end{description}
We first rearrange the above conditions.
Let $\mathbf{A}'=[\mathbf{A}|\mathbf{A}_1^{0}|\mathbf{A}_1^{1}|\cdots |\mathbf{A}_\ell^{0}|\mathbf{A}_\ell^{1}]\in\mathbb{Z}_q^{(2\ell+1)m}$, $\mathbf{A}_{\mathrm{id}}=[\mathbf{A}_0|\mathbf{A}_1^{\mathrm{id}[1]}|\cdots|\mathbf{A}_1^{\mathrm{id}[\ell]}]\in\mathbb{Z}_q^{(\ell+1)m}$ and $\mathbf{A}''=\mathbf{A}_{\ell+1}^{z[1]}|\cdots|\mathbf{A}_{\ell+1}^{z[d]}]\in\mathbb{Z}_q^{dm}$.
Then $\mathbf{A}_{\mathrm{id}\|z}=[\mathbf{A}_{\mathrm{id}}|\mathbf{A}'']\in\mathbb{Z}_q^{(\ell+d+1)m}$.
Let $\mathbf{v}_{\mathrm{id}}=(\mathbf{v}_0\|\mathbf{v}_1\|\cdots\|\mathbf{v}_{\ell})$, $\mathbf{w}_2=(\mathbf{v}_{\ell+1}\| \cdots\|\mathbf{v}_{\ell+d})$ with each $\mathbf{v}_i\in\mathbb{Z}^m$. Then
$\mathbf{v}_{\mathrm{id}\|z}=(\mathbf{v}_{\mathrm{id}}\|\mathbf{w}_2)$.
Therefore $\mathbf{A}_{\mathrm{id}\|z}\cdot \mathbf{v}_{\mathrm{id}\|z}=\mathbf{u}\bmod q$ is equivalent to
\begin{eqnarray}\label{equation:goal-1-1}
\mathbf{A}_{\mathrm{id} }\cdot \mathbf{v}_{\mathrm{id}}+\mathbf{A}''\cdot \mathbf{w}_2=\mathbf{u}\bmod q.
\end{eqnarray}
Since $\mathrm{id}$ is part of secret input, $\mathbf{A}_{\mathrm{id}}$ should not be explicitly given. We note that Langlois et al.~\cite{LLNW14-PKC} already addressed this problem. The idea is as follows: they first added $\ell$ suitable zero-blocks of size $m$ to vector $\mathbf{v}_{\mathrm{id}}$ and then obtained the extended vector $\mathbf{w}_1=(\mathbf{v}_0\|\mathbf{v}_1^0\|\mathbf{v}_1^1\|\cdots\|\mathbf{v}_{\ell}^0\|\mathbf{v}_{\ell}^1) \in\mathbb{Z}^{(2\ell+1)m}$, where the added zero-blocks are $\mathbf{v}_1^{1-\mathrm{id}[1]},\ldots,\mathbf{v}_\ell^{1-\mathrm{id}[\ell]}$ and $\mathbf{v}_i^{\mathrm{id}[i]}=\mathbf{v}_i,\forall i\in[\ell]$.  Now one can check that equation~(\ref{equation:goal-1-1}) is equivalent to
\begin{eqnarray}\label{equation:goal-1-2}
\mathbf{A}'\cdot \mathbf{w}_1+\mathbf{A}''\cdot \mathbf{w}_2=\mathbf{u}\bmod q.
\end{eqnarray}

Let $\mathbf{B}'=\left[
                   \begin{array}{ccc}
                    \mathbf{B}^\top& \mathbf{I}_m & \mathbf{0}^{m\times \ell} \\
                    \mathbf{G}^\top & \mathbf{0}^{\ell\times m} & \mathbf{I}_{\ell} \\
                   \end{array}
                 \right]$, $\mathbf{B}''=\left[
                                              \begin{array}{c}
                                                \mathbf{0}^{m\times\ell} \\
                                                \lfloor q/2\rfloor \mathbf{I}_{\ell} \\
                                              \end{array}
                                            \right]
                 $,  and  $\mathbf{w}_3=(\mathbf{s}\|\mathbf{e}_1\|\mathbf{e}_2)\in\mathbb{Z}^{n+m+\ell}$. Then one can check that  $\mathbf{c}_1=\mathbf{B}^\top\cdot\mathbf{s}+\mathbf{e}_1\bmod q, ~ \mathbf{c}_2=\mathbf{G}^\top\cdot\mathbf{s}+\mathbf{e}_2+\big\lfloor\frac{q}{2}\big\rfloor\cdot\mathrm{id}\bmod q$ is equivalent to \begin{eqnarray}\label{equation:goal-1-3}
\mathbf{B}'\cdot \mathbf{w}_3+\mathbf{B}''\cdot \mathrm{id}=(\mathbf{c}_1\|\mathbf{c}_2)\bmod q.
\end{eqnarray}
Using basic algebra, we can transform equations~(\ref{equation:goal-1-2})and~(\ref{equation:goal-1-3}) into one equation of the following form: \[\mathbf{M}_0\cdot \mathbf{w}_0=\mathbf{u}_0\mod q,\] where  $\mathbf{M}_0,~\mathbf{u}_0$ are built from  $\mathbf{A}',\mathbf{A}'',\mathbf{B}',\mathbf{B}''$ and $\mathbf{u},(\mathbf{c}_1\|\mathbf{c}_2)$, respectively, and $\mathbf{w}_0=(\mathbf{w}_1\|\mathbf{w}_2\|\mathbf{w}_3\|\mathrm{id})$.

Now we can use the decomposition and extension techniques described in Section~\ref{subsection:dec-ext} to handle our secret vectors.
Let $L_1=3(2\ell+1)mp_\beta$, $L_2=3dmp_\beta$, $L_3=3(n+m+\ell)p_B$, and $L=L_1+L_2+L_3+2\ell$.
We transform our secret vector $\mathbf{w}_0 $ to vector $\mathbf{w}=(\widehat{\mathbf{w}}_1\|\widehat{\mathbf{w}}_2\|\widehat{\mathbf{w}}_3\|\widehat{\mathrm{id}})\in\{-1,0,1\}^{L}$ of the following form:
\begin{itemize}
\item $\widehat{\mathbf{w}}_1=(\widehat{\mathbf{v}}_0\|\widehat{\mathbf{v}}_{1}^{0}\|\widehat{\mathbf{v}}_1^1\|\cdots\|\widehat{\mathbf{v}}_\ell^0\|\widehat{\mathbf{v}}_\ell^{1})\in\{-1,0,1\}^{L_1}$ with $\widehat{\mathbf{v}}_0=\mathsf{ext}_3(\mathsf{vdec}_{m,\beta}(\mathbf{v}_0))\in \mathsf{B}_{3mp_\beta}$, $\forall i\in[\ell]$, $\widehat{\mathbf{v}}_{i}^{1-\mathrm{id}[i]}=\mathbf{0}^{3mp_\beta}$ and $\widehat{\mathbf{v}}_{i}^{\mathrm{id}[i]}=\mathsf{ext}_3(\mathsf{vdec}_{m,\beta}(\mathbf{v}_{i}^{\mathrm{id}[i]}))\in\mathsf{B}_{3mp_\beta}$;  \smallskip
\item $\widehat{\mathbf{w}}_2=\mathsf{ext}_3(\mathsf{vdec}_{dm,\beta}(\mathbf{w}_2))\in \mathsf{B}_{3dmp_\beta}$;\smallskip
\item $\widehat{\mathbf{w}}_3=\mathsf{ext}_3(\mathsf{vdec}_{n+m+\ell,B}(\mathbf{w}_3))\in\mathsf{B}_{3(n+m+\ell)p_B}$;\smallskip
\item $\widehat{\mathrm{id}}=\mathsf{enc}_2(\mathrm{id})\in\{0,1\}^{2\ell}$.
\end{itemize}

Using basic algebra, we can form public matrix $\mathbf{M}$ such that \[\mathbf{M}\cdot \mathbf{w} = \mathbf{M}_0\cdot \mathbf{w}_0 =\mathbf{u}_0 \bmod q.\]

Up to this point, we have transformed the considered relations into equation of the desired form $\mathbf{M}\cdot \mathbf{w}=\mathbf{u}\bmod q$.
We now specify the set $\mathsf{VALID}$ that contains the secret vector $\mathbf{w}$, the set $\mathcal{S}$ and permutations $\{\Gamma_{\phi}: \phi\in \mathcal{S}\}$ such that the conditions in~(\ref{eq:zk-equivalence}) hold.

Define $\mathsf{VALID}$ to be the set of vectors of the form $\mathbf{z}=(\mathbf{z}_1\|\mathbf{z}_2\|\mathbf{z}_3\|\mathbf{z}_4)\in\{-1,0,1\}^{L}$ such that there exists $\mathbf{x}\in\{0,1\}^{\ell}$
\begin{itemize}
\item $\mathbf{z}_1=(\mathbf{y}_0\|\mathbf{y}_{1}^{0}\|\mathbf{y}_{1}^{1}\|\cdots\|\mathbf{y}_{\ell}^{0}\|\mathbf{y}_{\ell}^{1})\in\{-1,0,1\}^{3(2\ell+1)mp_\beta}$ with  $\mathbf{y}_0\in \mathsf{B}_{3mp_\beta}$  and for each $ i\in[\ell]$, $\mathbf{y}_{i}^{1-\mathbf{x}[i]}=\mathbf{0}^{3mp_\beta},~\mathbf{y}_{i}^{\mathbf{x}[i]}\in\mathsf{B}_{3mp_\beta} $; \smallskip

\item $\mathbf{z}_2\in \mathsf{B}_{3dmp_\beta} $ and $\mathbf{z}_3\in \mathsf{B}_{3(n+m+\ell)p_B} $; \smallskip

\item $\mathbf{z}_4=\mathsf{enc}_2(\mathbf{x})\in\{0,1\}^{2\ell}$. \smallskip
\end{itemize}
Clearly, our vector $\mathbf{w}$ belongs to the tailored set $\mathsf{VALID}$.\smallskip

Now, let $\mathcal{S}=(\mathcal{S}_{3mp_\beta})^{2\ell+1}\times \mathcal{S}_{3dmp_\beta}\times \mathcal{S}_{3(n+m+\ell)p_B}\times \{0,1\}^{\ell}$.  For any \[\phi=(\psi_0,\psi_1^0,\psi_1^1,\ldots,\psi_{\ell}^0,\psi_{\ell}^1,\eta_2,\eta_3,\mathbf{e})\in\mathcal{S},~\mathbf{e}=(e_1,\ldots,e_\ell)^\top,\] define the permutation $\Gamma_{\phi}: \mathbb{Z}^{L}\rightarrow\mathbb{Z}^{L}$ as follows. When applied to a   vector
\[\mathbf{z}=(\mathbf{y}_0\|\mathbf{y}_{1}^{0}\|\mathbf{y}_{1}^{1}\|\cdots\|\mathbf{y}_{\ell}^{0}\|\mathbf{y}_{\ell}^{1}\|\mathbf{z}_2\|\mathbf{z}_3\|\mathbf{z}_4)\in\mathbb{Z}^{L}\]
where the first $2\ell+1$ blocks are of size $3mp_\beta$ and  the last three blocks are of size $3dmp_\beta$, $3(n+m+\ell)p_B$ and $2\ell$, respectively;
it transforms $\mathbf{z}$ to   vector $\Gamma_{\phi}(\mathbf{z})$ of the following form:
\begin{align*}
(\hspace*{1.6pt}&\psi(\mathbf{y}_0)\|\psi_{1}^{e_1}(\mathbf{y}_1^{{e}_1})\|\psi_{1}^{1-e_1}(\mathbf{y}_1^{1-e_1})\|\cdots\|\psi_{\ell}^{e_{\ell}}(\mathbf{y}_\ell^{e_\ell})\|\psi_{\ell}^{1-e_\ell}(\mathbf{y}_\ell^{1-e_{\ell}})\|
\\
&\eta_{2}(\mathbf{z}_2)\|\eta_{3}(\mathbf{z}_3)\|\Pi_{\mathbf{e}}(\mathbf{z}_4)\hspace*{1.6pt}).
\end{align*}

Based on the equivalences observed in~(\ref{eqn:permutation}) and~(\ref{eqn:enc2}) , it can be checked that if  $\mathbf{z}\in \mathsf{VALID}$ for some $\mathbf{x}\in\{0,1\}^\ell$, then $\Gamma_{\phi}(\mathbf{z})\in \mathsf{VALID}$ for some $\mathbf{x}\oplus \mathbf{e}\in\{0,1\}^{\ell}$. In other words, the conditions in~(\ref{eq:zk-equivalence}) hold, and therefore, we can obtain the desired statistical $\mathsf{ZKAoK}$ protocol.

\section*{Acknowledgements}
The research is supported by Singapore Ministry of Education under Research Grant MOE2016-T2-2-014(S). Khoa Nguyen is also supported by the Gopalakrishnan -- NTU Presidential Postdoctoral Fellowship 2018.

\bibliographystyle{abbrv}

\appendix
\newpage
\section{Some Remarks on~\cite{KDM17} (ePrint 2017/1128)}\label{appendix:incorrect}
In~\cite{KDM17}, Kansal, Dutta, and Mukhopadhyay proposed a forward-secure group signature scheme from lattices in the model of Libert and Yung~\cite{LY10}. Unfortunately, it can be observed that their proposed scheme does not satisfy the correctness and security requirements.

The version of the scheme posted on 27-Nov-2017 15:26:21 UTC contains the following shortcomings.
\begin{itemize}
\item The scheme does not satisfy the correctness requirement.
    The opening algorithm, on input a signature generated by user $i$, does not output $i$. In fact, the transcript for user $i$ stored by the group manager is $\mathsf{transcript}_i=(\mathbf{v}_{i}^{(0)},i,\mathsf{uvk}[i],\mathsf{sig}_i,[t_1,t_2])$. However, when signing messages at time $t_1+j$ with $0< j \leq t_2-t_1$, the user $i$ encrypts $\mathbf{v}_i^{(j)}$, which is never seen by the group manager and which is unrelated to $\mathbf{v}_i^{(0)}$. Hence, the decryption procedure can only recovers $\mathbf{v}_i^{(j)}$ and no user is being traced in this case.  (Readers are referred to Page 22 (Join algorithm) and Page 25 (Sign algorithm) in~\cite{KDM17} for more details.)
         \smallskip
\item The scheme does not satisfy the anonymity requirement. The signature generated by user~$i$ at time period~$t_j$ contains matrix~$\mathbf{C}_i^{(j)}$, which is part of the updated certificate and which should be kept secret. Therefore, two signatures generated by the same user at the same period can easily be linked. (Readers are referred to  Page 24 (Update algorithm) and Page 27 (equation (13)) in~\cite{KDM17} for more details.)

    \smallskip

\end{itemize}

We note that in an updated version of the scheme, posted on 18-Jan-2018 17:33:37 UTC, the signature does not contain matrix $\mathbf{C}_i^{(j)}$, but the validity of the signature now cannot be publicly verified. That is because, in order to verify the underlying zero-knowledge argument system of Section~5 (Page 27), one needs to be given matrix $\mathbf{C}_{\mathrm{id}_i}^{(j)}$ that encodes the secret identity of the signer and that is not publicly known. In other words, this updated scheme also does not work.

We further observe that, a version of the scheme posted on 19-Apr-2018 07:26:41 UTC still contains the following shortcoming.

\begin{itemize}
\item The signatures generated by malicious users cannot be properly identified. In fact, to tackle with the first problem we have put forwarded above, the authors allow all users to update their transcripts to contain $\mathbf{v}_i^{(j)}$ at time $t_j$. However,  this enables all malicious user to  update the transcript arbitrarily. Therefore,  employing the transcript to identify the signer when opening signatures  is meaningless.
(Readers are referred to the updating algorithm in Page 25.)

\end{itemize}

\end{document}